\documentclass[journal, 12pt, draftclsnofoot, onecolumn]{IEEEtran}  

\IEEEoverridecommandlockouts                              
                         
\usepackage{graphics} 
\usepackage{epsfig} 
\usepackage{times} 
\usepackage{amsmath} 
\usepackage{amssymb}  
\usepackage{mathrsfs}
\usepackage{array}
\usepackage[linesnumbered]{algorithm2e}
\usepackage{color}

\usepackage{url}
\usepackage{array}
\usepackage[english]{babel}
\usepackage{times}
\usepackage{psfrag}
\usepackage{dsfont}
\usepackage{setspace}
\usepackage{amssymb}
\usepackage{amsmath}
\usepackage{graphicx}

\def\y{{\bf y}}
\def\x{{\bf x}}
\def\h{{\bf h}}
\def\Sigmab{{\bf \Sigma}}
\def\mub{{\bf \mu}}

\def\E{\,\mathds{E}\,}
\def\R{\mathds{R}}

\def\lcp{\left\{}
\def\rcp{\right\}}
\newtheorem{theorem}{Theorem}[section]
\newtheorem{corollary}[theorem]{ Corollary}
\newtheorem{lemma}[theorem]{Lemma}
\newtheorem{proposition}[theorem]{Proposition}

\renewcommand{\baselinestretch}{1.45}

\begin{document}

\title{Dynamic Optimization of Generalized \\ Least Squares Handover Algorithms}


\author{Carlo Fischione, George Athanasiou, Fortunato Santucci
\thanks{C. Fischione and G. Athanasiou are with the Automatic Control Lab, School of Electrical Engineering, KTH Royal Institue of Technology, Sweden. E-mail: \ttfamily{\{carlofi, georgioa\}@kth.se}}%
\thanks{F. Santucci is with the Centre of Excellence DEWS and DISIM, University of L'Aquila, L'Aquila, Italy.
                E-mail: \ttfamily{\{fortunato.santucci\}@univaq.it}}%
\thanks{The work of C.~Fischione and G. Athanasiou was supported by the Swedish Research Council and the EU projects Hycon2 and Hydrobionets. The work of F. Santucci was supported by EU project Hycon2 and a research contract with Thales Communications Italy.}
\thanks{A preliminary version of this work appeared in \cite{VTC07}. The authors thank C. Rinaldi and K. H. Johansson for discussions on background topics of this manuscript.}
}

\maketitle

\begin{abstract}
Efficient handover algorithms are essential for highly performing mobile wireless communications. 
These algorithms depend on numerous parameters, whose settings must be appropriately optimized to offer a seamless connectivity. Nevertheless, such an optimization is difficult in a time varying context, unless adaptive strategies are used. In this paper, a new approach for the 
handover optimization is proposed. First, a new modeling of the handover process by a hybrid system that takes as input the handover parameters is established. Then, this hybrid system is used to pose some dynamical optimization approaches where the probability of outage and the probability of handover are considered. Since it is shown that these
probabilities are difficult to compute, simple approximations of adequate accuracy are 
developed. Based on these approximations, a new approach to the solution of the handover optimizations is proposed by the use of a trellis diagram. A distributed optimization algorithm is then developed to maximize handover performance. From an extensive set of results obtained by numerical computations and simulations, it is shown that the proposed algorithm allows to improve performance of the handover considerably 
when compared to more traditional approaches.
\end{abstract}

\clearpage
\section{Introduction}

The handover process is the mechanism of transferring the connection between a
mobile station and a base station to another base station, so that
the communication can be maintained with adequate quality. The growth of cellular wireless systems with mobile communications, vehicular communications, and 
multi-protocols mobile terminals, has motivated the investigation of efficient handover algorithms that are able to offer a seamless connectivity, i.e., good quality of the communication precisely during the switching mechanism.

There is a long history of studies on handover since the beginning of the cellular era
\cite{sant99}\nocite{tripathi98,gud91,kumar94,veeravalli97,sant99,akar01,stojmenovic01,sant02,zhang02,akar03,AM_INFO2003,Fang04,kim07,diererich05,Lampropoulos05,Meriem08,Wang04,Choi+08,Sgora09,wang09,lee11,leu2004}~--~\cite{sung2003}. 
Handover algorithms can be classified in two types \cite{stojmenovic01}: soft and hard. 
In the hard handover, the connection to the serving
base station is released while the new base station takes on
the connection~\cite{zhang02}, \cite{leu2004}. In the soft handover, the mobile station can be
simultaneously connected to two or more different base stations~\cite{wang09}, \cite{sung2003}. This can
be achieved by exploiting the temporal diversity offered by
multi-path propagation (see, e.g., \cite{Choi+08}). Such a strategy gives
the smoothest connection and offers potentially higher performance
compared to the hard handover. Whereas the soft handover can be used generally between same wireless systems, hard handover can be used both between same system and between heterogeneous systems. However, soft handover is
highly expensive for the network operator. As a matter of fact,
recent standardization of LTE suggests the use of hard
handover only \cite{Sesia}. Moreover, hard handover seems the only option in inter-system handovers and, in general, in the emerging arena of connection management throughout heterogeneous wireless interfaces. Consequently, in this paper we restrict our
attention to such a class.

The papers \cite{sant99}\nocite{tripathi98,gud91,kumar94,veeravalli97,sant99,akar01,sant02,akar03,AM_INFO2003,kim07,lee11,diererich05,Lampropoulos05,Meriem08,Wang04,Sgora09}~--~\cite{Fang04} present handover approaches in the same system. They focus on the choice of the handover decision parameter, e.g., the received signal strength from the serving and
the neighboring base stations, the distance from base stations, and the bit error rate, within one wireless system. 
 In \cite{gud91}, the authors define three performance indicators: the probability of lost calls, the probability of unnecessary handover and the probability of handover, and they design a handover algorithm that trades-off among these indicators, since they cannot be minimized simultaneously. The way to adapt handover to the wireless environment, e.g.,
macro-cellular, micro-cellular, overlay systems, has been surveyed in 
\cite{tripathi98}. In \cite{kumar94}, the performance of a handover procedure using both bit error rate and relative signal strength measurements is analyzed. A call-quality criterion to balance against
the number of handovers in designing an optimal handover strategy is presented in \cite{veeravalli97}. Handover algorithms that are based on the least square (LS) estimate of path-loss parameters of the various radio links have been introduced in \cite{sant99}. In \cite{Wang04}, two schemes have been proposed for managing downlink CDMA radio resources that maintain ongoing call quality by minimizing call-dropping during handovers, without over-penalizing new arrivals. In both schemes, the guard capacity of a cell is dynamically adjusted so to maintain the handover dropping rate at or below a target level. 

The natural evolution of the papers surveyed above focuses on handovers between heterogeneous networks. In \cite{kim07}, the authors propose a vertical handover algorithm which is able to avoid the ping-pong effect which may occur when mobile terminals are moving between different networks. The proposed algorithm determines the appropriate time at which a terminal should initiate a handover. In addition, in \cite{Fang04} the authors propose new optimization techniques for 
handover decision with main target to maximize the benefit of the handover for both the user and the network. The optimizations incorporate a network elimination feature to reduce the delay and processing required in the evaluation of the cost function. A multi-network optimization is also introduced to improve throughput for mobile terminals with multiple active sessions. A group-handover approach is presented in \cite{lee11}, where an optimized network selection and adjusted delay in the initiation of handover to reduce the probability of handover blocking is considered. 

As the understanding of handover techniques and related performance has become more mature, more advanced mathematical tools have been proposed. The switching of the handover mechanism, both within the same wireless system or between different wireless systems, can be naturally modeled by a hybrid system framework, where a mixture
of continuous state and discrete event dynamics is taken into account. Hybrid systems provide a unified framework for describing processes
evolving according to continuous dynamics, discrete dynamics and
logic rules, e.g., ~\cite{mitra03}. Some of the most studied hybrid systems in the recent years
involve the interactions between
continuous and discrete dynamics, since this class of systems has
found applications in wireless
networks, embedded systems, and
control~\cite{hespanha01hybrid}.
In \cite{akar01}~--\nocite{akar03}~\cite{AM_INFO2003}, the idea to optimally control the switchings among base stations by hybrid systems is proposed. In \cite{akar03}, such an approach is extended by considering soft handover in a fading environment with interference. 

In this paper we propose a hybrid modeling and optimization of a general handover
algorithm. We consider the most important elements concurring in the handover mechanism, such as wireless channel estimation, performance selection, and handover optimization. We build upon the approach called Generalized Extended Least Squares handover (GELS) that was proposed in 
\cite{sant02}, where handover algorithms are designed based on the estimation of the path loss. In particular, we first propose a new mathematical model of GELS by a hybrid system where the different
time-varying dynamics of the wireless channel and handover decision are taken into account. Second, based on such a modeling, we propose a dynamic optimization approach to decide when performing the base station switching. Our
approach is related to relevant contributions in \cite{VTC07},
\cite{tripathi98}~--~\cite{AM_INFO2003}; however, our study differs significantly since we include in the
hybrid model the wireless channel estimation, various optimization problems, and solution algorithms, which was neglected before. We extend the previous work
in~\cite{VTC07, sant02} by tackling the challenging problem of the optimal decision rule for
handover. This requires an entirely new characterization of the handover performance indicators, such as probability of outage and probability of handover. In fact, we propose a new complete framework where the use of mathematical tools that include
models of the wireless propagation scenario and the system
characteristics are essential for offering highly
performant handover algorithms.

The rest of the paper is organized as follows. In Section~\ref{sec:model} the
basic system model for the handover is presented, whereas the
hybrid model of GELS handover scheme is described in Section~\ref{sec:gels} and expression of related performance metrics is investigated in Section~\ref{sec:indexes}. The handover 
optimization problems are presented in~\ref{sec:optprob}. Finally,
simulation results are presented in Section~\ref{sec:results}
and Section~\ref{sec:conclusions} concludes the paper.


\section{System Model}\label{sec:model}

We consider a general scenario for the handover given by a mobile terminal (MT) 
moving among a number of Base Stations (BSs). The MT measures the signal
strengths coming from all the BSs. Suppose that there are $S$ base
stations. Let the received signal strength in a log-scale from each BS be
\begin{align*}
p_s(n)=\alpha_s-\beta_s\textrm{log}[d_s(n)]+u_s[d_s(n)],  \quad s=0,1,\ldots,S-1
\quad n=1, \ldots N 
\end{align*}
where $d_s(n)$ is the distance at time instant $n$ between the MS and BS $s$, $\alpha_s-\beta_s\textrm{log}[d_s(n)]$ is the path loss,
whereas $u_s[d_s(n)]$ models the shadow fading. The terms
$u_i[d_i(n)]$ and $u_j[d_j(n)]$ for $i\neq j$ are assumed to be zero
mean Gaussian processes independent of each other. Such a model of the received signal strength is common in the handover literature, e.g., \cite{AM_INFO2003}. It results from an average over a number of samples, that are typically spaced enough in the time domain so that the fast fading is filtered out. 

For the handover algorithm, the estimate of the signal strength received from BS$_s$, from the
sequence $\{p_s(n)\}$, $i=1,2,\ldots,n$, is calculated as follows \cite{sant99}, \cite{sant02}:
\begin{align*}
l_s(n)=\sum_{i=n_b}^{n}p_s(i)G_s(n,i) \,,
\end{align*}
with $n_b=\max\{1,n-n_w+1\}$, where $n_w$ is the length of the
window  used for the estimation, and $G_s(n,i)$ are filtering
coefficients. Both the window length and the filtering coefficients depend on the handover algorithm. 

The active BS is defined as the BS to which a mobile terminal is connected. The best BS is defined as the BS with the best signal strength among those neighboring BSs that are candidates for the handover connection. Let the subscripts $0$ and $1$ denote two involved BSs, say the active BS and the best BS at a certain time instant, respectively,
and let $l_s(n)$ be the linear estimate of the signal strength
measured from BS$_i$, with $s=0$ or $1$. Consider the following random variable:
\begin{align} \label{eq:y}
y(n)=l_0(n)-l_1(n) \,.
\end{align}
The handover decision is based on the comparison of $y(n)$ with a
hysteresis margin $h(n)$. In particular, if MT is connected to a base station, say BS$_0$, it disconnects from it if $y(n) \leq -h(n)$, 
whereas if MT is connected to base station BS$_1$, it disconnects from it if $y(n) \geq h(n)$.
This mechanism is illustrated in Figure~\ref{fig:hysteresis}. 

\begin{figure}[t]
\centering
\includegraphics[width=0.5 \textwidth]{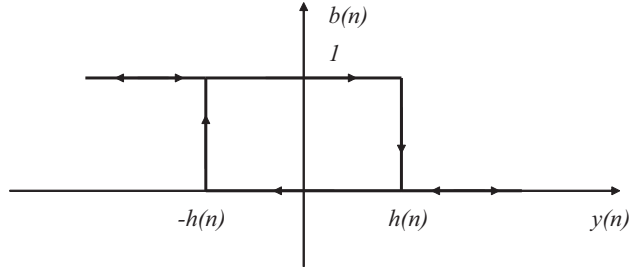}
\caption{The hysteresis of the handover decision. In this paper, the hysteresis margin $h(n)$ is selected for each time instant $n$ to optimize the handover.}
\label{fig:hysteresis}
\end{figure}

Based on this mechanism, during a certain time interval it is possible that multiple handovers happen, when the MT connects to a new BS at a time instant $n$ and then disconnects from a serving BS at the following time instant $n+1$. To ensure a high quality in the communication, it is desirable that the MT stays connected to the BS offering high quality in the communication, i.e., low outage probability. However, due to the mobility of the MT and the resulting variation of the wireless channel, especially in the outer cell region, it is uncertain which BS offers the best performance and therefore, disconnection and reconnection can happen. These disconnections and reconnections are quite expensive from an operator point of view, which has to transfer the MT state information from one BS to the other. It is therefore important to avoid as much as possible frequent reconnections. Such a trade off between MT and BS can be regulated by appropriately selecting the handover hysteresis. Ultimately, the handover performance is measured in terms of outage probability and handover probability, which depend on 1) the estimation process that gives the $l_0(n)$, $l_1(n)$ and $y(n)$; 2) the appropriate selection of the handover hysteresis margin $h(n)$.
In the following, we characterize these two essential aspects, and provide a method for the optimal handover decision.





\section{Hybrid Model of the Handover Algorithm}\label{sec:gels}

In this section, we present a new modeling of the estimation process that provides us with $l_s(n)$ and $y(n)$. The estimation is based on GELS estimation. The new modeling uses the hybrid system theory formalism, which will be included in the optimization process to select the handover hysteresis margin. To introduce such a model, we need some details on intermediate estimators 
that are part of the hybrid system: the AVG (Averaging), LS (Least Squares), and ELS (Extended Least Squares) estimators. We describe first the intermediate estimators in the sequel, according to \cite{sant99}, \cite{sant02}. 


The simplest way of computing the estimator coefficients
$G_s(n,i)$ is given by a simple average, which is the AVG
algorithm. The algorithms selects the coefficients as
$G_s(n,i)=g(n-i)$, where $g(n)$ is a given filter impulse
response. A simple choice is a rectangular shape for $g(n)$.
Accordingly, AVG is a simple filtering of the measured power
$p_s(\cdot)$. We turn our attention in the following to a more advanced estimator. 


LS makes the estimate $l_s(n)$ by assuming a model of the channel
attenuation and minimizing the squares of the difference between
the power of the received signal and the power of the assumed model, namely it is
assumed that the received power follows the model
$l_s(n)=\hat{\alpha}_s(n)-\hat{\beta}_s(n)\log d_s(i)$, where
$\hat{\alpha}_s(n)$ and $\hat{\beta}_s(n)$ have to be estimated by
a least square minimization over $n_w$ samples. Thus, the function that we would like to minimize is
\begin{align*}
\sum_{i=n_b}^{n}(p_s(i)-\hat{\alpha}_s(n)+\hat{\beta}_s(n)\log
d_s(i))^2 \,.
\end{align*}
The minimization yields
\begin{align*}
\hat{\alpha}_s(n)&=\frac{1}{D_s(n)-C_s^2(n)}\cdot[P_s(n)D_s(n)-Q_s(n)C_s(n)] \\
\hat{\beta}_s(n)&=\frac{1}{D_s(n)-C_s^2(n)}\cdot[P_s(n)C_s(n)-Q_s(n)] \,,
\end{align*}
and
\begin{align*}
P_s(n)& \triangleq \frac{1}{n-n_b+1}\sum_{i=n_b}^{n}p_s(i)\,, \qquad Q_s(n) \triangleq \frac{1}{n-n_b+1}\sum_{i=n_b}^{n}p_s(i)\log{d_s(i)}\\
C_s(n)& \triangleq \frac{1}{n-n_b+1}\sum_{i=n_b}^{n}\log{d_s(i)}\,, \qquad D_s(n) \triangleq \frac{1}{n-n_b+1}\sum_{i=n_b}^{n}(\log{d_s(i)})^2\,.
\end{align*}
From the minimization, it follows that
\begin{align}
l_s(n)=\sum_{i=n_b}^{n}p_s(i)A_s(n,i)-p_s(i)B_s(n,i)\log{d_s(n)}\,,
\end{align}
where
\begin{align*}
A_s(n,i)& \triangleq  \frac{1}{n-n_b+1}
\frac{D_s(n)-C_s(n)\log{d_s(i)}}{D_s(n)-C_s^2(n)} \qquad 
B_s(n,i) \triangleq
\frac{1}{n-n_b+1}\frac{C_s(n)-\log{d_s(i)}}{D_s(n)-C_s^2(n)}\,,
\end{align*}
whereby we see that the filter coefficients are
\begin{align}
G_s(n,i)&=A_s(n,i)-B_s(n,i)\log{d_s(n)} \,.
\end{align}
In the following the AVG and LS estimated are combined to provide the ELS estimator. 


The ELS algorithm schedules the use of both AVG and LS estimates,
depending on the reliability of the estimates of path loss
parameters provided by AVG and LS. ELS adapts to channel changing
conditions by comparing the errors on path loss estimates of AVG
and LS, which are indicated by $e_{s,1}$ and $e_{s,2}$, respectively, and
then choosing the estimate with the lower error~\cite{sant99,sant02}. 



GELS is a generalized version of ELS. It relies on
successive steps, starting from the 'easiest to handle' linear
handover algorithm, which is based on the averaging (AVG) of the signal
strength, up to operation in the most complex scenario with an adaptation mechanism for the
hysteresis. GELS uses the following variables:
\begin{align*}
&\Delta_s(n)=p_s(n)-l_s(n)\,, \quad e_{s,\min}(n)=\min\{e_{s,1}(n),e_{s,2}(n)\}\,, \quad &e_{s,r}(n)=\frac{\Delta_s(n)}{\sqrt{e_{s,\min}(n)}} \,.
\end{align*}
The last variable provides a measure of the validity of the estimator
model with respect to the current path loss model actually experimented on the channel. The GELS algorithm works as
follows~\cite{sant02}: the variables $e_{s,\min}(n)$, $l_s(n)$, $\Delta_s(n)$,
and $e_{s,r}(n)$ are computed for each BS$_s$-MT link, at time
instant $n$. If $\Delta_s(n)>>e_{s,r}(n)$ the estimator model is
not valid anymore and the estimates on each link are
re-initialized. The hysteresis margin $h(n)$ is intended to be adapted and an estimation in re-initialization 
of a link can be forced if $h(n)>h_M$, where
$h_M$ is the maximum value allowed for the hysteresis margin.


We are now ready to give the core contribution of this section. Consider the current base station, BS$_0$, to which MT is connected, and the BS with the strongest signal, BS$_1$. By the description above, we can model the GELS algorithm by a discrete
time hybrid system:
\begin{align} \label{Sn1}
S(n+1)=A(n)S(n)-f(d(n+1), d(n))+W(n)\,,
\end{align}
where:
\begin{align} \label{SN}
S(n)=[p_0(n), p_1(n), l_0(n), l_1(n)]^T\,,
\end{align}
\begin{align}
A(n)=\left[
\begin{array}{cccc}
1 & 0 & 0 & 0\\
0 & 1 & 0 & 0\\
G_0(n+1) & 0 & 1 & 0\\
0 & G_1(n+1) & 0 & 1
\end{array}
\right]\,,
\end{align}
\begin{align}
f(d(n+1), d(n))=\left[
\begin{array}{c}
\beta_0 \log\left(\frac{d_0(n+1)}{d_0(n)}\right)\\
\beta_1 \log\left(\frac{d_1(n+1)}{d_1(n)}\right)\\
G_0(n+1)\beta_0 \log\left[\frac{d_0(n+1)}{d_0(n)}\right]\\
G_1(n+1)\beta_1 \log\left[\frac{d_1(n+1)}{d_1(n)}\right]
\end{array}
\right]\,,
\end{align}
\begin{align}
W(n)=\left[
\begin{array}{c}
u_0(n+1)-u_0(n)\\
u_1(n+1)-u_1(n)\\
G_0(n+1)[u_0(n+1)-u_0(n)]\\
G_1(n+1)[u_1(n+1)-u_1(n)]
\end{array}
\right] \,,
\end{align}
and
\begin{align} \label{dn}
d(n)=\left[d_0(n) \: \: \: d_1(n)\right]^T\,.
\end{align}
Thus, the switching between the BSs after taking a handover decision 
can be described by the update equation for the discrete
variable $b(n) = \mathscr{E}(n)$, 
where $\mathscr{E}(n)$ is the event occurring when the MT is
connected to BS$_1$ at time $n$, namely
$\mathscr{E}(n)=\{b(n)=1\}$. Analogously, we let $\mathscr{\bar
E}(n)=\{b(n)=0\}$. In practice, $\mathscr{E}(n)$ is 1 when the MT is connected to the
BS$_1$, and $0$ when it is connected to the BS$_0$. The event $\mathscr{E}(n)$ can be expressed as
follows:
\begin{align} \label{epsilon} 
\mathscr{E}(n)  =  \{y(n) < - h(n)\} + \{y(n) < h(n),
\mathscr{E}(n-1)\}
\end{align}
and its complementary event is
\begin{align} \label{notepsilon}
\mathscr{\bar E}(n) =  \{y(n) \geq h(n)\} + \{y(n) \geq -h(n),
\mathscr{\bar E}(n-1)\} \,.
\end{align}

Eqs.~\eqref{Sn1} --~\eqref{dn} define a hybrid system with $h(n)$ as an input variable, which determines
the values of $\mathscr{E}(n)$. 
The next natural question is how to select such a variable so that the dynamic of the hybrid system, namely the handover process, is optimized with respect to some performance indicators. To answer this question, we need to characterize these indicators, as we will see in the following section.



\section{Expression of Performance Metrics}\label{sec:indexes}

We choose the probability of outage and the probability of handover as the performance metrics for the handover process. Specifically, we argument and derive the expressions of these probabilities within a quite general framework, that encompasses e.g. a hysteresis margin different from zero. These expressions will then be used in the next section to carry out optimization of the handover algorithm subject to the estimation of the channel. 
 With this goal in mind,
first we derive the probability of MT being connected to each of the base stations,
namely the probability of the event $\mathscr{E}(n)$.

\subsection{Probability of Connection to a Base Station}

The following events concur in the definition of Eq.~\eqref{epsilon} and Eq.~\eqref{notepsilon}:
\begin{align}
\mathscr{L}(h(n)) & = \{y(n) \leq -h(n)\}\,, \\
\mathscr{M}(h(n)) & = \{-h(n)<y(n)\leq h(n)\}\,, \\
\mathscr{N}(h(n)) & = \{y(n)>h(n)\} \,.
\end{align}
Then, we get the following results:
\begin{proposition} \label{prop:epsilon}
Consider the events corresponding to the MT being connected to the base station 1 and 0,
namely $\mathscr{E}(n)=\{b(n)=1\}$ and $\mathscr{\bar
E}(n)=\{b(n)=0\}$. Let $0\leq m \leq n$, where $n$ is the discrete time index. Then
\begin{align}
 \mathscr{E}(n) = & \sum_{j=m+1}^n \mathscr{L}(h(j)) \prod_{k=j+1}^n \mathscr{M}(h(k)) + \prod_{k=m+1}^n \mathscr{M}(h(k)) \mathscr{E}(m) \,, \label{eq:En}\\
\mathscr{\bar E}(n)  = & \sum_{j=m+1}^n \mathscr{N}(h(j)) \prod_{k=j+1}^n \mathscr{M}(h(k)) + \prod_{k=m+1}^n \mathscr{M}(h(k)) \mathscr{\bar E}(m)
\label{eq:notEn} \,.
\end{align}
\end{proposition}

\begin{IEEEproof}
Consider the definition of the
event $\mathscr{E}(n)$ given in Eq.~\eqref{epsilon}. Since
$\mathscr{\bar N}(h(n)) \subseteq \mathscr{L}(h(n))$, it follows
that
\begin{align*}
\mathscr{E}(n) = & \mathscr{L}(h(n)) + \mathscr{\bar N}(h(n))\mathscr{E}(n-1) = \mathscr{L}(h(n)) + \mathscr{M}(h(n))\mathscr{E}(n-1) \\
= & \mathscr{L}(h(n)) + \mathscr{M}(h(n))\mathscr{L}(n-1) + \mathscr{M}(h(n))\mathscr{M}(h(n-1))\mathscr{E}(n-2) \,,
\end{align*}
where the last inequality is obtained by using the definition of
$\mathscr{E}(n-1)$ given by Eq.~\eqref{epsilon} at time $n-1$. By
repeating the procedure until $0\leq m \leq n$ we obtain
sought expression~\eqref{eq:En}. The same argument can be applied to
obtain Eq.~\eqref{eq:notEn}, which concludes the proof.
\end{IEEEproof}
\begin{corollary} \label{coro:prob-epsilon}
Consider the events of MT being connected to the base station 1 and 0,
namely $\mathscr{E}(n)=\{b(n)=1\}$ and $\mathscr{\bar
E}(n)=\{b(n)=0\}$. Let $0\leq m \leq n$. Then
\begin{align}
\Pr[\mathscr{E}(n)]= & \sum_{j=1}^n \Pr \lcp \mathscr{L}(h(j)) \prod_{k=j+1}^n \mathscr{M}(h(k)) \rcp + \Pr\lcp \prod_{k=1}^n \mathscr{M}(h(k)) \rcp \mathscr{E}(0) \,, \label{eq:PrEn}\\
\Pr[\mathscr{\bar E}(n)]=& \sum_{j=1}^n \Pr\lcp \mathscr{N}(h(j)) \prod_{k=j+1}^n \mathscr{M}(h(k))\rcp + \Pr\lcp \prod_{k=1}^n \mathscr{M}(h(k)) \rcp \mathscr{\bar
E}(0) \label{eq:PrnotEn} \,.
\end{align}
\end{corollary}
\begin{IEEEproof}
The proof results from Proposition~\ref{prop:epsilon} by setting
$m=0$ and observing that $\mathscr{E}(n)$ and $\mathscr{\bar
E}(n)$ are given by the sum of mutually exclusive events, which concludes the proof. 
\end{IEEEproof}

We use the characterization of connection to the base stations given by
Proposition~\ref{prop:epsilon} and Corollary~\ref{coro:prob-epsilon} in subsections~\ref{sec:handover-prob} and~\ref{sec:outage-prob} to
characterize the probability of handover and the probability of outage. However, in order to get those expressions, it is instrumental to perform the explicit computation of the probabilities given by Proposition~\ref{prop:epsilon} and Corollary~\ref{coro:prob-epsilon}. 

The computation of the probabilities~\eqref{eq:PrEn}
and~\eqref{eq:PrnotEn} relies on a multivariate Gaussian
distribution, since the events $\mathscr{L}(k)$,
$\mathscr{M}(h(k))$, and $\mathscr{N}(k)$ are defined over
Gaussian cross correlated random variables. In particular, by using the expressions of filter coefficients derived in Section III, we see that the
random variable $y(n)$ defined in Eq.~\eqref{eq:y} has a Gaussian
distribution with average
\begin{align}
\mu_{y}(n)  =  & \E\{y(n)\}  = \sum_{i=n_b}^{n}
[\alpha_0-\beta_0\log d_0(i)]G_0(k,i) -\sum_{i=n_b}^{n}[\alpha_1-\beta_1\log d_1(i)]G_1(k,i) \,,
\end{align}
and covariance matrix $\Sigma=[\Sigma_{nh}]$, with $
\Sigma_{nh}=\rho_y(n,n-h)\sigma_y(n)^2\sigma_y(n-h) $, 
where
\begin{align} \sigma_y^2(n)  = & \E\left\{[y(n)-\mu_y(n)]^2\right\} \nonumber \\
= &
\sum_{i=n_b}^{n}\sum_{j=n_b}^{n}r_{u_0}(i-j)G_0(n,i)G_0(n,j) + \sum_{i=n_b}^{n}\sum_{j=n_b}^{n}r_{u_1}(i-j)G_1(n,i)G_1(n,j)\,, \label{sigmay}
\end{align}
and
\begin{align} \label{eq:corr}
r_{u_s}(l)\triangleq
\tilde{r}_{u_s}(lvT)=\sigma_{u_s}^2e^{-|l|vT/\bar{d}}
\end{align}
is the autocorrelation function of $u_s(n)$ (which accounts for
the memory of the shadowed wireless channel, and the more compact notation $u_s(n)$ is used instead of $u_s[d_s(n)]$), $n_b=\max\{1,n-n_w+1\}$, and
$$
\rho_y(n,n-h)\triangleq \frac{\E
\left\{[y(n)-\mu_y(n)][y(n-h)-\mu_y(n-h)]\right\}}{\sigma_y(n)\sigma_y(n-h)}
\,.
$$
Furthermore,
\begin{align*}
& \E\left\{[y(n)-\mu_y(n)][y(n-h)-\mu_y(n-h)]\right\} \nonumber \\
& \qquad = \sum_{i=n_b}^{n}\sum_{j=n'_b}^{n-m}\left[r_{u_0}(i-j)G_0(n,i)G_0(n-m,j) +r_{u_1}(i-j)G_1(n,i)G_1(n-m,j)\right]
\end{align*}
where $n'_b=\max\{1,n-m-n_w+1\}$.

However, the use of a multivariate Gaussian distribution may be computationally
prohibitive, especially when the sizes of the vector grows. Thus, we propose to an approximation of the probabilities~\eqref{eq:PrEn}
and~\eqref{eq:PrnotEn} using expressions of the adequate accuracy and the reduced computational complexity. 

In the following subsections, we propose some useful upper bounds and lower bounds for the probability of the events expressed over Gaussian random vectors.

\subsection{Approximation 1} \label{subsec:ulb1}

Consider the Gaussian random vector $\y \in \R^k$ $\y=[y(1), y(2),\ldots, y(k)]^T$. 
Let $\mu_{\y} =  [\mu_k] \in \R^k$ be the average of $\y$, and let
$\Sigma=[\Sigma_{hk}]\in R^{k \times k}$ the covariance matrix.
The computation of the probability of events defined over a correlated Gaussian vector
may be prohibitive when the size of the vector is large, because
of the large number of integrations. Since we are interested on 
the use of these probabilities for optimization purposes, it is
natural to resort to approximations that give upper and lower bounds.

We propose the use of the following simple approximation, where random
variables are grouped into subsets. For each subset, the precise
probability is computed. The approximation is based on the
definition of the size of the subsets and on taking the product of
the probabilities (thus leading to independent events when different subsets are concerned) :
\begin{align}
& \Pr \lcp \mathscr{Y}_1 \mathscr{Y}_2 \ldots \mathscr{Y}_k \rcp  \approx \Pr \lcp \mathscr{Y}_{k} \mathscr{Y}_{k-1}
\mathscr{Y}_{k-3} \mathscr{Y}_{k-4}\rcp \times \cdots \times \Pr
\lcp \mathscr{Y}_{4}\mathscr{Y}_{3}\mathscr{Y}_{2}\mathscr{Y}_{1}
\rcp \,.
\end{align}

The validity of this approximation is discussed in Section\ref{sec:results}.

\subsection{Approximation 2} \label{subsec:ulb2}

Here we propose general bounds. They are computationally
simple but, given their generality, may be not accurate in all circumstances.
The following lemma gives a useful intermediate result. 
\begin{lemma}~\label{lemma:gulb}
Consider the Gaussian vector $\y \in \R^k$ having average $\mub$
and covariance matrix $\Sigmab$. Let $\lambda_{\rm max}$ and
$\lambda_{\rm min}$ be the maximum and minimum eigenvalue of
$\Sigmab$, respectively. Consider the sets $\mathscr{Y}_l= \lcp
y_l \in [\underline{y}_l , \bar{y}_l] \rcp$ for $l=1\ldots k$. Then
\begin{align} \label{eq:gub}
& \Pr \lcp  \mathscr{Y}_1 \mathscr{Y}_2 \ldots \mathscr{Y}_k \rcp  \leq  \frac{\lambda_{\max}^{\frac{k}{2}}}{\sqrt{\det \Sigmab}} \prod_{l=1}^{k} \Pr \lcp y_l \in \left[\frac{\sigma_l\underline{y}_l}
{\sqrt{\lambda_{\max}}}, \frac{\sigma_l\bar{y}_l}{\sqrt{\lambda_{\max}}}\right] \rcp \,, \\
& \Pr \lcp  \mathscr{Y}_1 \mathscr{Y}_2 \ldots \mathscr{Y}_k \rcp  \geq  \frac{\lambda_{\min}^{\frac{k}{2}}}{\sqrt{\det \Sigmab}} \prod_{l=1}^{k} \Pr \lcp  y_l \in \left[\frac{\sigma_l\underline{y}_l}
{\sqrt{\lambda_{\min}}}, \frac{\sigma_l\bar{y}_l}{\sqrt{\lambda_{\min}}}\right] \rcp \label{eq:glb}\,.
\end{align}
\end{lemma}
\begin{proof}
For every $\x \in \R^k$ it holds that 
$
\|\x\|^2 / \lambda_{\max} \leq \x^T \mathbf{\Sigma}^{-1} \x \leq \frac{\|\x\|^2}{\lambda_{\min}} \,.
$
Therefore 
\begin{align*}
& \Pr \lcp \mathscr{Y}_1  \mathscr{Y}_2 \ldots \mathscr{Y}_k \rcp  \nonumber \\
& = \int_{\underline{y}_1}^{\bar{y}_1} \int_{\underline{y}_2}^{\bar{y}_2} \ldots \int_{\underline{y}_k}^{\bar{y}_k} \frac{e^{-\frac{1}{2}
( \mathbf{y}-\mathbf{\mu})^T\mathbf{\Sigma}^{-1}(\mathbf{y}-\mathbf{\mu})}}{\sqrt{\det \mathbf{\Sigma}(2 \pi)^2}}dy_1 \ldots d y_k \leq \int_{\underline{y}_1}^{\bar{y}_1} \int_{\underline{y}_2}^{\bar{y}_2} \ldots \int_{\underline{y}_k}^{\bar{y}_k} \frac{e^{-\frac{1}{2}
\frac{\|\mathbf{y}-\mathbf{\mu}\|^2}{\lambda_{\max}}}}{\sqrt{\det \mathbf{\Sigma}(2 \pi)^2}}dy_1 \ldots d y_k\,,
\end{align*}
whereby~\eqref{eq:gub} follows after simple algebra. The
derivation of~\eqref{eq:glb} is given by a similar argument.
\end{proof}

From the previous lemma we observe that, if the matrix $\mathbf{\Sigma}$ is
well conditioned, then the upper bound~\eqref{eq:gub} and lower
bound~\eqref{eq:glb} are consistent, since the ratio
$\lambda_{\max}/\lambda_{\min}$ will be small. Actually, the
covariance matrix is expected to be well conditioned because of the limited
memory of the wireless channel. More precisely, let $m$ be such a
memory in terms of number of discrete time instants. Then the
elements of the covariance matrix that are more than $m$ locations
before and after the diagonal have negligible values. By
applying the Gersgoring theorem \cite{HornJohnson85},
we see that
\begin{align*}
\min_i \left( \Sigma_{ii}-\sum_{j=l_i, j \neq i}^{u_i}\Sigma_{ij} \right)  \leq \lambda_{\min} \leq \lambda_{\max} \leq \max_i \left( \Sigma_{ii}+\sum_{j=l_i, j \neq i}^{u_i}\Sigma_{ij} \right) \,,
\end{align*}
where $l_i=\max(0,i-m)$ and $u_i=\min(k,i+m)$. Given the
correlation pattern~\eqref{eq:corr}, the sum of the off diagonal
elements of $\mathbf{\Sigma}$ is expected to be small with respect to
$\Sigma_{ii}$, which implies small conditioning numbers.

\subsection{Approximation 3} \label{subsec:ulb3}

Here we develop a bound that is more computationally
demanding, but exhibits better accuracy.
We use the following intermediate result:
\begin{lemma} \label{fact:approxjointprob}
Let $\mathscr{A}$ and $\mathscr{B}$ be two correlated events. Then
\begin{equation}
\Pr \lcp \mathscr{A} \mathscr{B} \rcp \leq \sqrt{\Pr \lcp
\mathscr{A} \rcp } \cdot \sqrt{\Pr \lcp \mathscr{B} \rcp } \,.
\end{equation}
\end{lemma}
\begin{IEEEproof}
Since $\Pr \lcp \mathscr{A} \mathscr{B} \rcp \leq \Pr \lcp
\mathscr{A} \rcp$ and $\Pr \lcp \mathscr{A} \mathscr{B} \rcp \leq
\Pr \lcp \mathscr{B} \rcp$, it follows that $\Pr ^2 \lcp
\mathscr{A} \mathscr{B} \rcp \leq \Pr \lcp \mathscr{A} \rcp \Pr
\lcp \mathscr{B} \rcp$.
\end{IEEEproof}
We use previous simple result in the following proposition.

\begin{proposition}
Let $\y \in \R^k$ be a Gaussian vector having average $\mub$ and
covariance matrix $\Sigmab$. Suppose $m\leq k$ and let
$\Sigmab_{k-m}$ be the matrix obtained by taking the first $k-m$ rows
and $k-m$ columns of $\Sigmab$. Let $\lambda_{k-m, \max}$ be the
maximum eigenvalue of $\Sigmab_{k-m}$. Consider the subsets
$\mathscr{Y}_l= \lcp y_l \in [\underline{y}_l , \bar{y}_l] \rcp$ for
$l=1\ldots k$. Then
\begin{align}
& \Pr \lcp \mathscr{Y}_1 \mathscr{Y}_2 \ldots \mathscr{Y}_k \rcp  \\
& \quad \leq \sqrt{ \Pr \lcp \mathscr{Y}_{k} \mathscr{Y}_{k-1} \ldots \mathscr{Y}_{k-m+1} \rcp } \frac{\lambda_{k-m,\max}^{\frac{k-m}{4}}}{\det^{1/4} \Sigmab_{k-m}}  \times \prod_{l=1}^{k-m} \Pr^{1/2} \lcp y_l \in \left[\frac{\sigma_l\underline{y}_l}
{\sqrt{\lambda_{k-m, \max}}}, \frac{\sigma_l\bar{y}_l}{\sqrt{\lambda_{k-m, \max}}}\right] \rcp \,. \nonumber
\end{align}
\end{proposition}
\begin{proof}
From Lemma~\ref{fact:approxjointprob} 
\begin{align}\label{eq:PrG}
& \Pr \lcp \mathscr{Y}_1 \mathscr{Y}_2 \ldots \mathscr{Y}_k \rcp \leq \sqrt{ \Pr \lcp  \mathscr{Y}_{k} \mathscr{Y}_{k-1}
\ldots \mathscr{Y}_{k-m+1} \rcp } \sqrt{ \Pr \lcp \mathscr{Y}_1
\mathscr{Y}_2 \ldots \mathscr{Y}_{k-m} \rcp } \,.
\end{align}
By applying Lemma~\ref{lemma:gulb} to the second probability of the
right end-side of previous inequality, the proposition easily follows.
\end{proof}

We are now in the position to derive expression for the probabilities of handover and outage, respectively. 

\subsection{Probability of Handover}\label{sec:handover-prob}

In this subsection, we provide an expression of the probability of
handover. The following result holds:
\begin{proposition} \label{prop:handover-prob}
Consider the serving BS and the strongest candidate BS. The probability of handover at time $n$ is
\begin{align} \label{eq:PH}
P_H(n) = & P_{H01}(n) + P_{H10}(n) \,,
\end{align}
where
\begin{align}
P_{H01}(n) & = \Pr[\mathscr{N}(n)\mathscr{E}(n-1)]\,, \\
P_{H10}(n) & = \Pr[\mathscr{L}(h(n))\mathscr{\bar E}(n-1)]\,.
\end{align}
\end{proposition}
\begin{IEEEproof}
The occurrence of handover events can be described by the following iterative expression $
\mathscr{H}(n) =  \mathscr{H}_{01}(n) + \mathscr{H}_{10}(n) =
\mathscr{E}(n)\mathscr{\bar E}(n-1) + \mathscr{\bar E}(n)
\mathscr{E}(n-1)$\,.
Consider the event $\mathscr{H}_{01}$, then
\begin{align*}
\mathscr{H}_{01}(n) = & \mathscr{E}(n)\mathscr{\bar E}(n-1) =  \mathscr{N}(n)\mathscr{L}(n-1)+ \mathscr{N}(n)\mathscr{\bar N}(n-1)\mathscr{E}(n-1) \\
=& \mathscr{N}(n)\mathscr{L}(n-1)+ \mathscr{N}(n)\mathscr{M}(h(n-1))\mathscr{E}(n-1) = \mathscr{N}(n)\mathscr{L}(n-1)+ \mathscr{N}(n)\mathscr{E}(n-1) \\
=& \mathscr{N}(n)\mathscr{E}(n-1)\,,
\end{align*}
where the third equality turns out by observing that $\mathscr{\bar N}(n-1)
\subseteq \mathscr{L}(n-1)$, and the last equality results from
$\mathscr{N}(n)\mathscr{L}(n-1) \subseteq
\mathscr{N}(n)\mathscr{E}(n-1)$.

By following the same arguments, $ \mathscr{H}_{10}(n) = \mathscr{\bar E}(n)\mathscr{E}(n-1) =
\mathscr{L}(h(n))\mathscr{\bar E}(n-1)$\,.

Notice that $\mathscr{H}_{10}(n)$ and $\mathscr{H}_{01}(n)$ are
mutually exclusive. Therefore the proposition follows.
\end{IEEEproof}

As observed for calculation of the
probabilities~\eqref{eq:PrEn} and~\eqref{eq:PrnotEn}, it may be
challenging to compute the
probability~\eqref{eq:PrG}
 by the exact Gaussian multivariate distribution. Hence, we can use the approximations
proposed in the previous subsections.


\subsection{Outage Probability} \label{sec:outage-prob}

In this subsection we derive the expression of the probability of
outage. The events of outage when MT is connected to the serving and to the strongest candidate BS are defined as 
\begin{align}
& \mathscr{P}_0(n)=\{p_0(n) \leq \beta\} \,, \\
& \mathscr{P}_1(n)=\{p_1(n) \leq \beta\} \,.
\end{align}
Then we have the following result:
\begin{proposition} \label{prop:outage-prob}
The outage probability at time $n$ is
\begin{align} \label{eq:PH}
P_O(n) = & P_{O0}(n) + P_{O1}(n) \,,
\end{align}
where
\begin{align}
P_{O0}(n) = \Pr[\mathscr{P}_0(n) |\mathscr{\bar E}(n)] = \frac{\Pr[\mathscr{P}_0(n) \mathscr{\bar E}(n)]}{\Pr[\mathscr{\bar E}(n)]}\,, \\
P_{O1}(n) = \Pr[\mathscr{P}_1(n)| \mathscr{E}(n)] =
\frac{\Pr[\mathscr{P}_1(n)\mathscr{E}(n)]}{\Pr[\mathscr{E}(n)]}\,.
\end{align}
\end{proposition}
\begin{IEEEproof}
The occurrence of the outage events is described by
\begin{align*}
\mathscr{O}(n)& =\mathscr{O}_0(n)+\mathscr{O}_1(n) = \mathscr{P}_0(n) |\mathscr{\bar E}(n) +
\mathscr{P}_0(n)|\mathscr{E}(n)\,,
\end{align*}
from which the proposition follows immediately by considering that
$\mathscr{O}_0(n)$ and $\mathscr{O}_1(n)$ are mutually exclusive.
\end{IEEEproof}

As for the calculation of the
probabilities of handover, it may be quite expensive to
compute the probability of outage by the Gaussian multivariate distribution. Hence, we use the
approximations proposed in the previous subsections. The accuracy of these approximations/bounds is discussed in Section VI. 

Now that we have characterized the expressions of the base station probability, outage probabilities, and handover probabilities we can turn our attention to the optimization of the handover. 

\section{Handover Optimization}\label{sec:optprob}

As for the hybrid system model represented by Eqs.~\eqref{Sn1} --~\eqref{dn}, the GELS algorithm uses a hysteresis margin for the handover decision. In a dynamic environment, the performance of the handover can be enhanced by selecting the hysteresis margin according to adaptive optimality criteria.
%
In the following, we propose three optimization
criteria that are based on dynamic programming \cite{bertsekas95}. 

The first strategy proposes the minimization of the probability of handover, while keeping under control the outage probability. This approach relies on the alreary mentioned rationale that completing a handover process is expensive due to the costs of transferring the connection from one base station to another one. Thus, it is beneficial to minimize the probability of handover as long as the probability of outage stays below a threshold. A second optimization approach proposes the reverse: the outage probability is minimized while the handover probability is kept under control. This is especially important for communications that need to have the highest successful packet reception probability, since fewer outage events allow to improve the successful bit decoding rates. Finally, the third approach proposes the minimization of the weighted combination of the two probabilities by a Pareto optimization method. The tradeoff between outage probability and handover probability is consistent, as it is typically observed that the outage probability increases and the handover probability decreases with the hysteresis margin.
The dynamic programming nature of the optimization problems is accounted for the definition of the cost function, which considers future evolutions of the hysteresis thresholds. In the following, we present the three methods. 

\subsection{Probability of Handover Optimization}

In this subsection we propose the optimization of the handover
probability under outage constraints. More specifically, here we
investigate the following dynamic optimization problem:
\begin{subequations} \label{eq:phstpo}
\begin{align}
\min_{\h(n)} \quad & \sum_{l=n}^{n+m} P_{H b(l)}(l)  \label{objeq:phstpo}\\
{\rm s.t.} \quad & P_{Ob(l)}(l) \leq P_{\rm out}\,, \quad l=n,\ldots,n+m \nonumber \\
& b(l+1)=\mathscr{E}(n+1) \,, \quad l=n,\ldots,n+m  \\
& S(l+1)=A(l)S(l)-f(d(l+1), d(l))+W(l) \quad  l=n,\ldots,n+m
\end{align}
\end{subequations}
In such a problem, the decision variables are the hysteresis
thresholds $h(l)$ for $l=n,\ldots,n+m$, which we collect in the
vector $\h(n)=[h(n) \ldots h(n+m)]^T$. Note that $S(l)$ is given in Eq.~\eqref{SN}, the probability of outage is given by Proposition~\ref{prop:handover-prob}, and the probability of Handover is given by Proposition~\ref{prop:outage-prob}. 

At each time instant $n$, the mobile station tries to
minimize the probability of handover over a time window that spans from the current time instant up
to a future instant that is $m$ sampling times ahead of $n$. The
handover probability is minimized while taking into account outage
events, which motivates the outage probability constraint for
ensuring an adequate quality of the communication. In other words,
we impose that at each time instant $l$, $l=n,\ldots,n+m$, the
outage probability must be below a maximum value $P_{\rm out}$.
The last constraint of the optimization problem returns the BS 
 $b(l+1)$ at which the MS is connected to at
time $l+1$ when a hysteresis threshold $h(l)$ is decided at time
$l$. Such a mobile station will then determine computation of
the handover probability $P_{H b(l+1)}(l+1)$ at time $l+1$.

Such an optimization involves a prediction of future
evolutions of the wireless channel. The memory of the channel is
finite owing to the coherence time \cite{proakis}. That is why a
prediction can be efficiently done over a finite time interval $m$. The dynamic
optimization that we are proposing is
motivated by observing that choosing a hysteresis threshold $h(n)$
at time $n$ determines the handover decisions and outage events of
the future times. Therefore, an optimization of the handover
looking just at a present time may have negative consequences in
the future, and needs to be done dynamically. 

In case that $m=1$, it is easy to show that it is a Fast-Lipschitz optimization problem~\cite{Fischione11} and thus very easy to solve. When $m\neq 1$, the problem becomes more complex. The difficulty arises by the fact that it is not convex due to the non-convexity of the cost function and by the fact that the selection of the optimal $h(l)$ affects the selection of $h(l+1)$, $h(l+2)$, and so on, due to the switching mechanism between BSs. We propose later in Subsection~\ref{solutionmethod} an algorithm to solve that problem.

\subsection{Probability of Outage Optimization}

Here we pose the optimization problem
\begin{subequations} \label{objfunc}
\begin{align}
\min_{\h(n)} \quad & \sum_{l=n}^{n+m} P_O(l) \label{objeq:objfunc} \\
{\rm s.t.} \quad & P_H(l) \leq P_{\rm han} \quad l=n,\ldots,m  \\
& b(l+1)=\mathscr{E}(n+1)  \quad l=n,\ldots,m \\
& S(l+1)=A(l)S(l)-f(d(l+1), d(l))+W(l) \quad l=n,\ldots,n+m
\end{align}
\end{subequations}
where the objective is the minimization of the outage probability subject to that the handover probability is limited by a maximum threshold $P_{\rm han}$. 
The decision variables are the hysteresis
thresholds $h(l)$ for $l=n,\ldots,n+m$, which are collected in the
vector $\h(n)=[h(n) \ldots h(n+m)]^T$ as for the previous optimization problem. The optimization takes into account the future evolution of the outage probability, because a handover decision taken at the current time $n$ will affects future events of the outages due to the switching of the BS. 

The solution of this optimization problem faces the same challenges as the problem~\eqref{eq:phstpo} does. Therefore, we follow the approach presented in Subsection~\ref{solutionmethod} to solve the optimization problem. Next, we propose a problem formulation that combines the previous two optimization problems. 

\subsection{Handover and Outage Pareto Optimization}

A more complex approach consists in solving an optimization problem where the
objective function is defined in terms of both outage and handover
probabilities:
\begin{equation} \label{Jn}
J(b(n),h(n))=\sum_{l=n}^{n+m} z \cdot P_H(l)+ (1-z)P_O(l) \:,
\end{equation}
where $P_H(l)$ is the handover probability at time $l$, $P_O(l)$
is the outage probability at time $l$, $z$ is a weighting
coefficient to tradeoff the performance in terms of outages or
handovers, and $m$ is the time horizon. The objective function is
therefore a weighted sum of handover and outage probability. In
the notation adopted for the cost function, we have evidenced the
dependance on the hysteresis $h(n)$ and the base station $b(n)$ at
which the mobile station is connected to. Thus, we can formulate the
following optimization problem
\begin{subequations} \label{objfunct}
\begin{align}
\min_{\h(n)} \quad & J(b(n),h(n)) \label{objeq:objfunct} & \\
& b(l+1)=\mathscr{E}(n+1)  \quad l=n,\ldots,m \\
& S(l+1)=A(l)S(l)-f(d(l+1), d(l))+W(l) \quad  l=n,\ldots,n+m
\end{align}
\end{subequations}
When this optimization problem is compared to the previous two problems, it is obvious that it is even more difficult to solve due to the complexity of the cost function. In the next subsection, we propose a solution algorithm for the problems \eqref{eq:phstpo}, \eqref{objfunc} and \eqref{objfunct}.


\begin{figure}[t]\vspace{-18mm}
\centering
\psfrag{ controller}{\ Optimization }
\psfrag{ estimator}{\ Estimator}
\psfrag{ system}{\ System}
\psfrag{hysteresis}{$h_n$}
\psfrag{distance}{$d_s(n)$}
\psfrag{power}{$p_s(n)$}
\psfrag{coefficient}{$G_s(n)$}
\psfrag{state}{$x(n)$}
\includegraphics[width=0.8\textwidth]{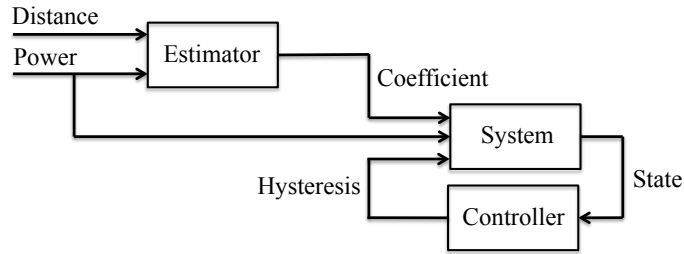} \vspace{-18mm}
\caption{Optimization scheme. For every time instant, an estimation of the channel coefficients is followed by an optimization of the hysteresis margin. The application of the optimal margin will determine the next connection to a Base Sation} \label{control}
\end{figure}

\subsection{Solution Method} \label{solutionmethod}
In this subsection, we propose an algorithm to solve in practice optimization problems \eqref{eq:phstpo}, \eqref{objfunc} and \eqref{objfunct}. 

Given the dynamic programming nature of the optimization with a binary variable (the base station) and a real variable (the hysteresis margin), 
we propose the use of an algorithm based on a trellis diagram (as depicted in Fig~\ref{control}). Specifically, every
stage of the trellis is associated to a time span from $n$ till $m$. At time $n$, the trellis has one state corresponding to the current base station $b(n)$. For the time instants $n+1$, $n+2$, $\ldots$ $n+m$, there are a number of possible states corresponding to one of the base stations the MT can be connected to. The transition from the base station at time $n$ to one of the next base stations at time $n+1$ has associated a probability of handover or a probability of outage. In Figs.~\ref{fig:trellis-handover} and~\ref{fig:trellis-outage}, we report two examples of trellis for the case of two base stations and $m=4$. The optimization algorithm works by the use of the trellis as follows:

\begin{enumerate}

\item for every path starting from $b(n-1)$ and ending to
one of the possible values of $b(n+m)$, the objective
function is computed as a function of the hysteresis. Depending on problem \eqref{eq:phstpo}, \eqref{objfunc} and \eqref{objfunct}, the objective function will be given by 
\eqref{objeq:phstpo}, \eqref{objeq:objfunc}, and \eqref{objeq:objfunct}, respectively. 

\item for every path, the hysteresis value that minimizes the cost function
corresponding to that path is computed; 

\item once the hysteresis values are known, it is possible to compute numerically the objective
function associated to every path, and thus the actual cost path;

\item the path with the minimum numerical objective function gives the value of $b(n+1)$ and thus next base station; 

\item the trellis diagram goes to the next state, when a new value of the fading parameters is produced. The
trellis is updated by removing the last stage, and adding a new one.
\end{enumerate}

In the following section, we illustrate the application of this algorithm and provide numerical results. 



\begin{figure}[t]
  \centering
  \includegraphics[width=0.6 \textwidth]{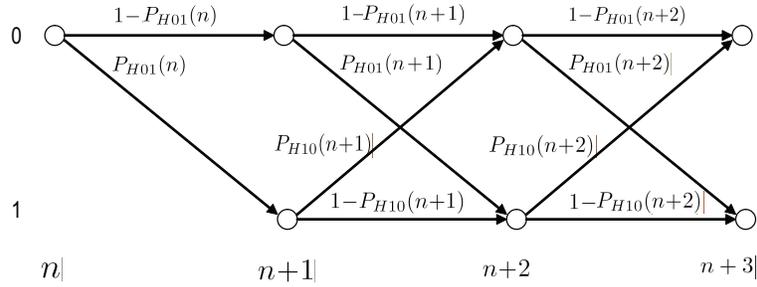}
  \caption{Trellis diagram for the handover probabilities in the case of $m=4$ and two base stations.}
  \label{fig:trellis-handover}
  \end{figure}
  
\begin{figure}[t]\vspace{-12mm}
  \centering
  \includegraphics[width=0.6 \textwidth]{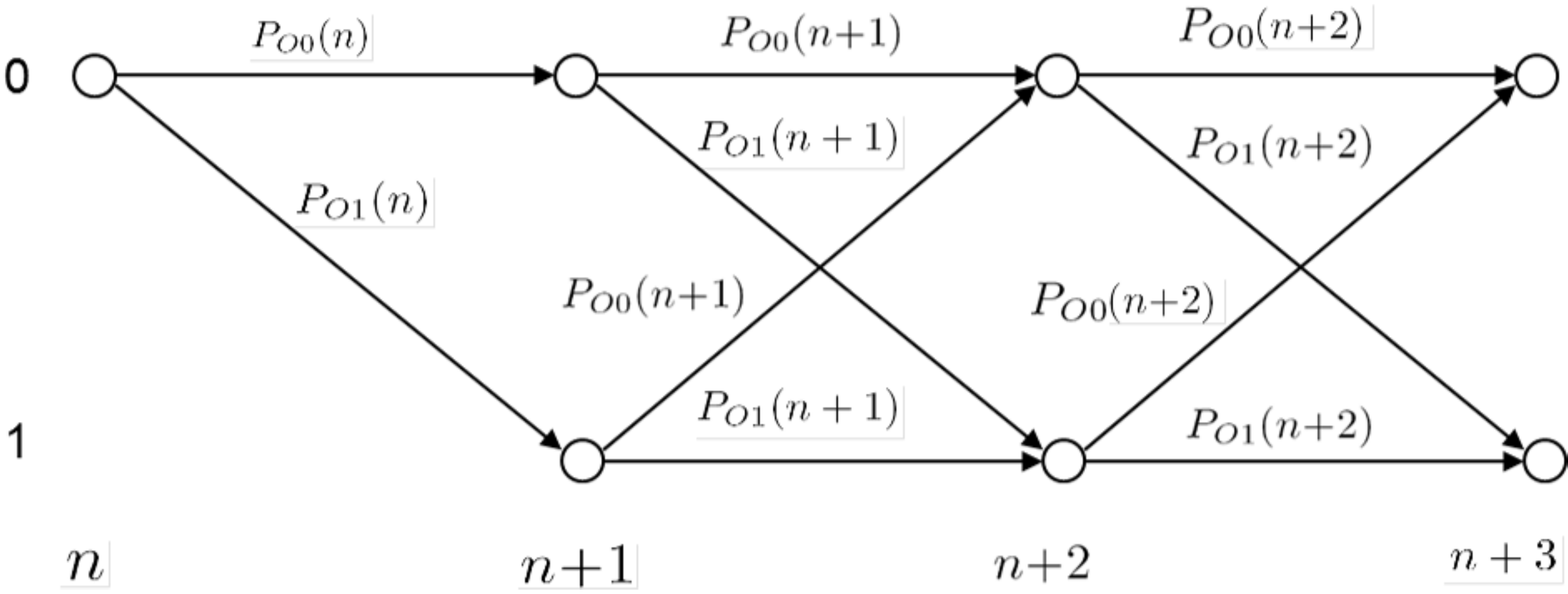}\vspace{-14mm}
  \caption{Trellis diagram for the outage probabilities in the case of $m=4$ and two base stations.}
  \label{fig:trellis-outage}
  \end{figure}

\section{Numerical results}\label{sec:results}

In this section we present the evaluation study of the adaptive handover approach in both a two-cells and in a multi-cell wireless environment and we compare our results to an existing method in literature.

First, we solve the optimization problems of Section \ref{sec:optprob} based on the trellis algorithm of Subsection~\ref{solutionmethod}. We make appropriate use of the approximations proposed above to reduce computational cost towards the calculation of the optimal $h(n)$ values. We then present an extensive study of the accuracy of these approximations. Lastly, we perform a simulation study in a multi-cell environment.

We start by describing the system settings used in the simulation runs. We assume that the MT is moving along a straight line towards the cell boundary. Since we mainly refer to vehicular communications over roads, the cells are assumed to have a nominal radius of $1\:$Km. Recall that there is need of handover in a region which is close to the cell boundaries. Thus, we take into account a path of total distance of $500\:$m starting at $750\:$m far from BS$_0$. The coherence interval of the wireless (shadowed) channel is assumed to be $\bar{d}=20$ m, which implies that predicted values of the wireless channel coefficients are actually effective only up to $20$ meters far from the starting point. This means that, if the current BS and the hysteresis value are known at time $n-1$, the future values of these parameters can
be predicted up to the time instant $n+3$. In fact, assuming a
standard sampling distance $d_c=v \cdot T_c=6.24 \textrm{m}$, where
$v=13$m/s and $T_c=0.48$s are the speed of the MT and the sampling
interval, respectively, we see that the number of the prediction
stages (and thus the number of the stages of the trellis diagram) is $ \bar{d}/d_c=4$.

In the following, we present some numerical results. 

\subsection{Optimal Hysteresis and Probabilities of Handover/Outage}
In this section we consider the simplest case where a MT is moving between two cells (from BS$_0$, Base Station 0, towards BS$_1$, Base Station 1). Our main target is to solve the optimization problems \eqref{eq:phstpo}, \eqref{objfunc} and \eqref{objfunct} through dynamic programming by the trellis structure as described in Subsection~\ref{solutionmethod}. In this way we get the optimal hysteresis threshold $h(n)$ at time $n$, which minimizes the objective functions, defined in the optimization problems.

Fig.~\ref{fig:opth1} plots the optimal $h(n)$ values that we get from optimization problem~\eqref{eq:phstpo}. The computation of the probabilities of handover and outage may be computationally expensive (since the multivariate Gaussian distribution may be computationally prohibitive). Therefore, we compare the optimal $h(n)$ values that we get by analytically computing the handover probability with those values that we obtain by applying the approximations presented in Section IV. We observe that $h(n)$ is kept at low levels when the MT is close to a BS. The optimal hysteresis margin grows as the MT moves towards the cell boundaries, in order to avoid unnecessary handovers in the system. The optimal $h(n)$ values that result when the approximations are applied is very close to the optimal $h(n)$ values resulted when the exact analytical computation of the probabilities is adopted.  

\begin{figure}[t]
  \centering
  \includegraphics[width=0.6 \textwidth]{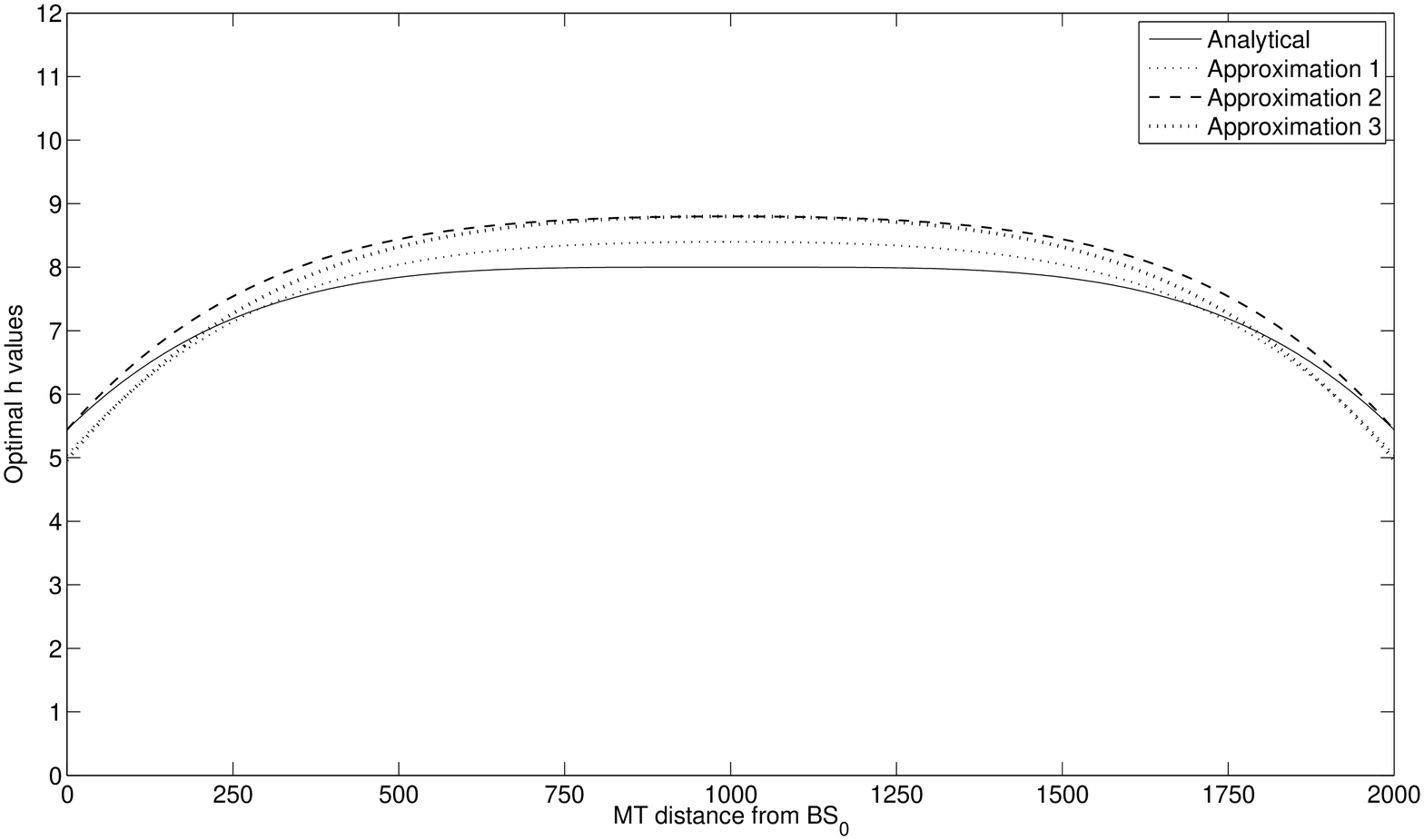}
  \caption{Optimal $h(n)$ values for optimization problem \eqref{eq:phstpo}.}
  \label{fig:opth1}
  \end{figure} 
  
\begin{figure}[t]
  \centering
  \includegraphics[width=0.6 \textwidth]{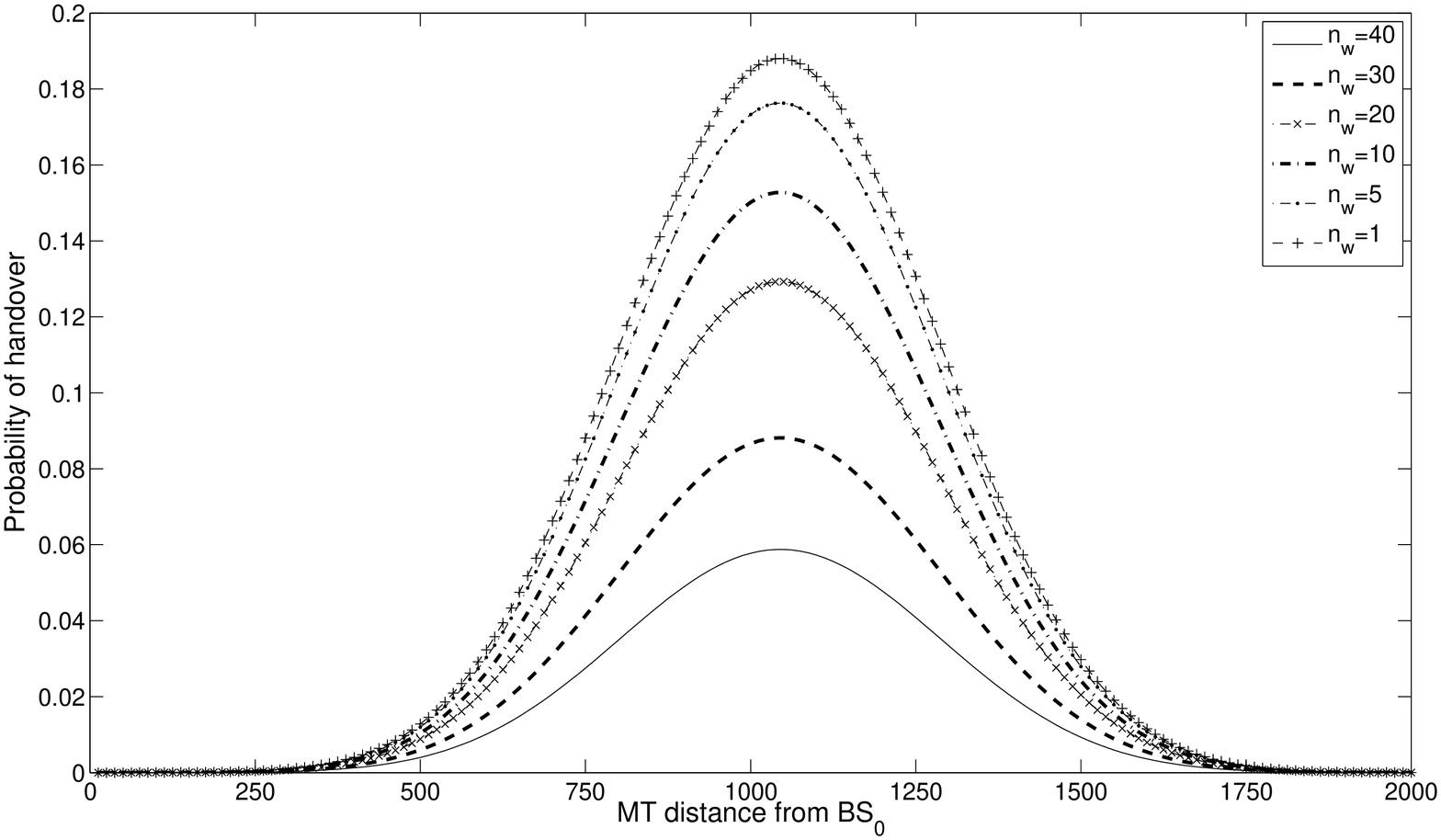}
  \caption{Optimal values of the probability of handover resulted from the solution of optimization problem~\eqref{eq:phstpo}.}
  \label{fig:P_h_opth1}
  \end{figure} 
  
We now study the behavior of the probability of handover when the window length $n_w$ is varied in the system. Fig.~\ref{fig:P_h_opth1} depicts the probability of handover while the MT moves towards BS$_1$. It is evident that the probability of handover decreases as the window length increases. This happens due to the reduced shadowing fluctuations that is guaranteed when large window length is used in the estimation process from the received signal strength. Moreover, we compare the probability of handover that is computed when constant hysteresis margins are used with the probability that is computed when the optimal $h(n)$ values are applied. Fig.~\ref{fig:P_opth1a} shows that higher hysteresis margins result in lower handover probabilities. The handover probability is minimized when the optimal $h(n)$ values are applied.  

   \begin{figure}[t]
  \centering
  \includegraphics[width=0.6 \textwidth]{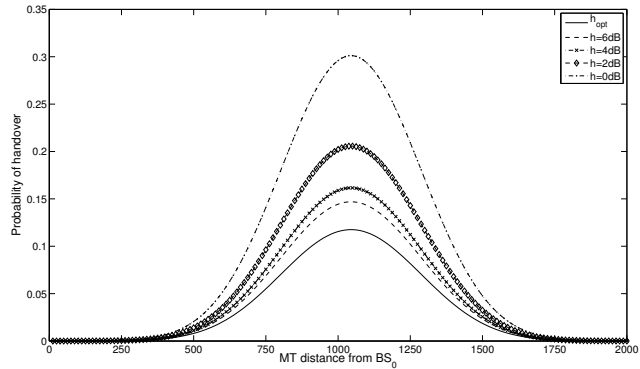}
  \caption{Optimal values of the probability of handover resulted from the solution of optimization problem~\eqref{eq:phstpo} (constant $h(n)$ values vs optimal).}
  \label{fig:P_opth1a}
  \end{figure} 
  
 \begin{figure}[t]
  \centering
  \includegraphics[width=0.6 \textwidth]{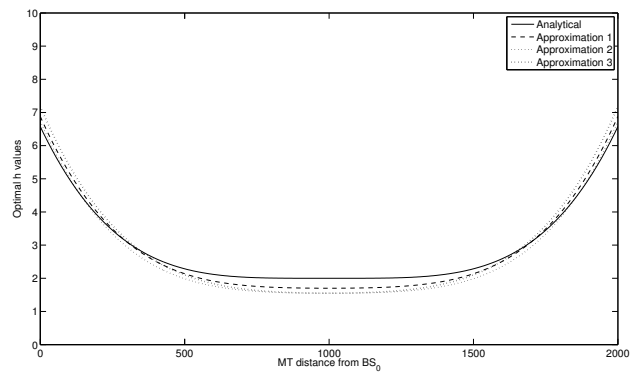}
  \caption{Optimal $h(n)$ values for optimization problem~\eqref{objfunc}.}
  \label{fig:opth2}
  \end{figure}
  
Fig.~\ref{fig:opth2} plots the optimal $h(n)$ values that we get for \eqref{objfunc}. We observe that the hysteresis margin decreases while the MT moves towards the cell boundaries, both when analytical computation of the outage probability is applied and when the approximations are used. $h(n)$ increases when the MT is close to the base stations. A general outcome here is that the behavior of $h(n)$ in \eqref{eq:phstpo} is in contrast to the behavior of the optimal hysteresis margin resulted from \eqref{objfunc}: as long as we are interested on minimizing outages, handovers are not prevented when the MT is much closer to a BS.

\begin{figure}[t]
  \centering
  \includegraphics[width=0.6 \textwidth]{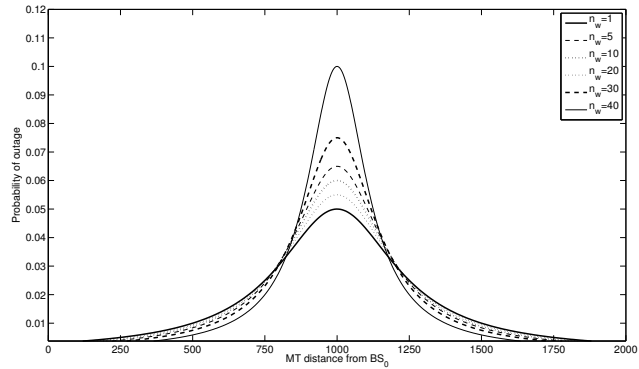}
  \caption{Optimal values of the probability of outage resulted from the solution of optimization problem~\eqref{objfunc}.}
  \label{fig:Pout_opt2_nw}
  \end{figure} 
  
  \begin{figure}[t]
  \centering
  \includegraphics[width=0.6 \textwidth]{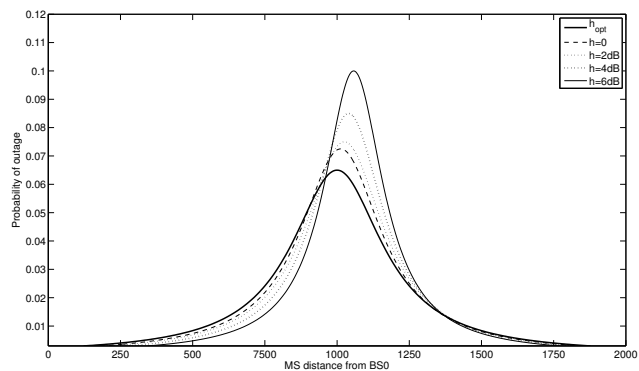}
  \caption{Optimal values of the probability of outage resulted from the solution of optimization problem~\eqref{objfunc} (constant $h(n)$ values vs optimal).}
  \label{fig:Pout_opt2_h}
  \end{figure} 
  
It is interesting to measure now the probability of outage while we vary the estimation window length $n_w$ and the hysteresis margin $h(n)$ in our system. From Figs. \ref{fig:Pout_opt2_nw} and \ref{fig:Pout_opt2_h} we observe that the behavior of the outage probability is in contrast to the behavior of the handover probability (similar effect to the variation of the hysteresis margin). We get higher outage probabilities values when high $n_w$ values are used. Moreover, as the hysteresis margin increases, the outage probability gets larger and the maximum shift to the right (handover delay when fixed hysteresis is used). Therefore, it is obvious that a trade-off exists in the minimization of the handover probability and the probability of outage, that must be controlled by the correct adaptation of the system parameters to achieve balanced network operation. 

Eq. \eqref{objfunct} tries to manage the existing trade-off by adapting $z$ according to the system characteristics and the application requirements. Fig.~\ref{fig:opth3} plots the optimal $h(n)$ values resulted from \eqref{objfunct} while varying the values of $z$. In addition, in Fig.~\ref{fig:opth3a} we observe the computation of the optimal handover and the outage probabilities, parametrized with respect to the weight coefficient $z$. We can see that large values of average handover probability result when $z$ is close to 0 (in that case the probability of outage is low). On the contrary, we get large average outage probability values when $z$ is close to 1 (in that case the probability of handover is low). A "knee" is present in the region where $z=0.6$. In order to efficiently manage the existing trade-off, there is a desire to operate close to that "knee" (by adapting $z$, $h(n)$ and $n_w$).

\begin{figure}[t]
  \centering
  \includegraphics[width=0.6 \textwidth]{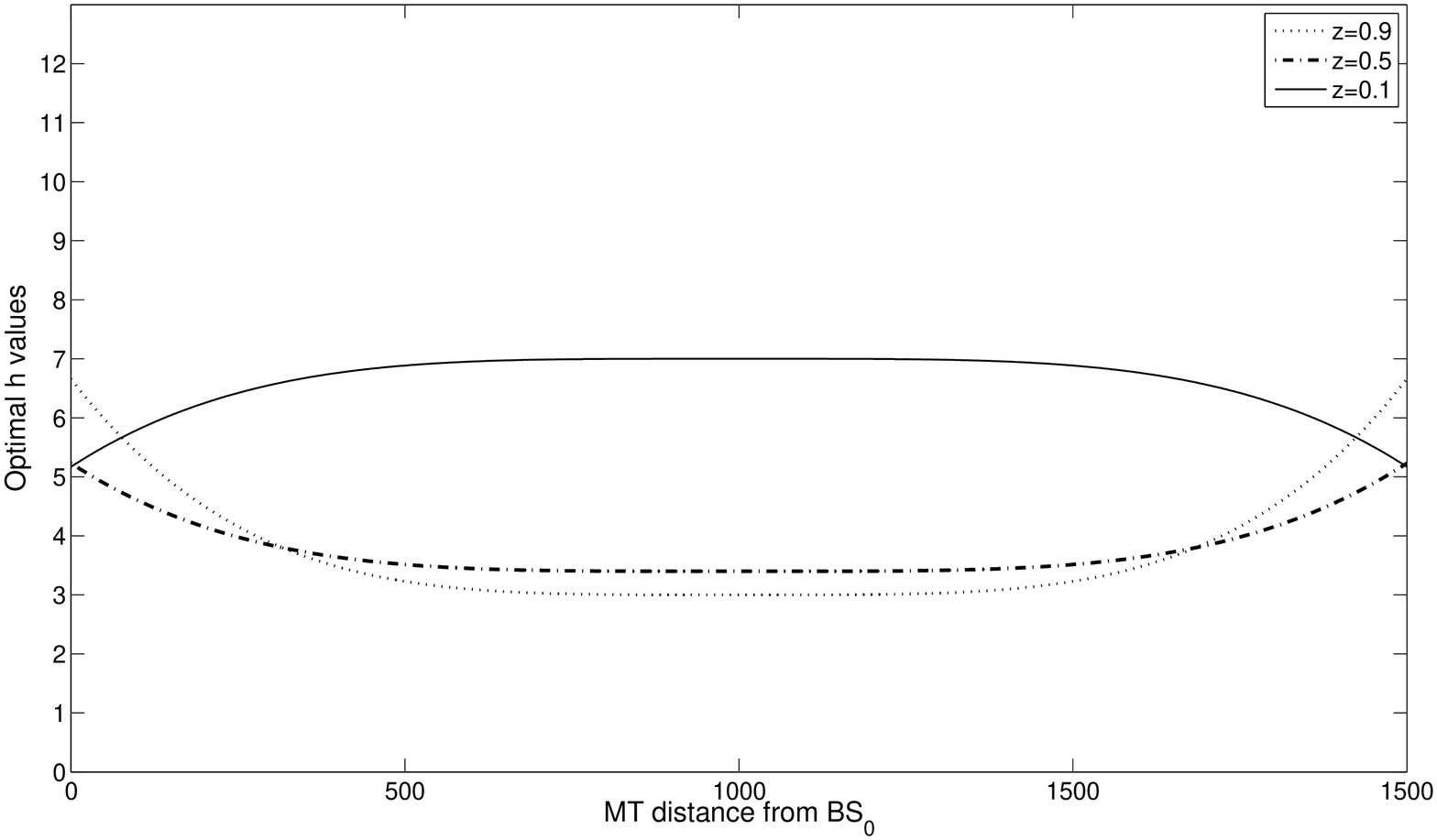}
  \caption{Optimal $h(n)$ values for the third optimization problem.}
  \label{fig:opth3}
  \end{figure}
  
\begin{figure}[t]
  \centering
  \includegraphics[width=0.6 \textwidth]{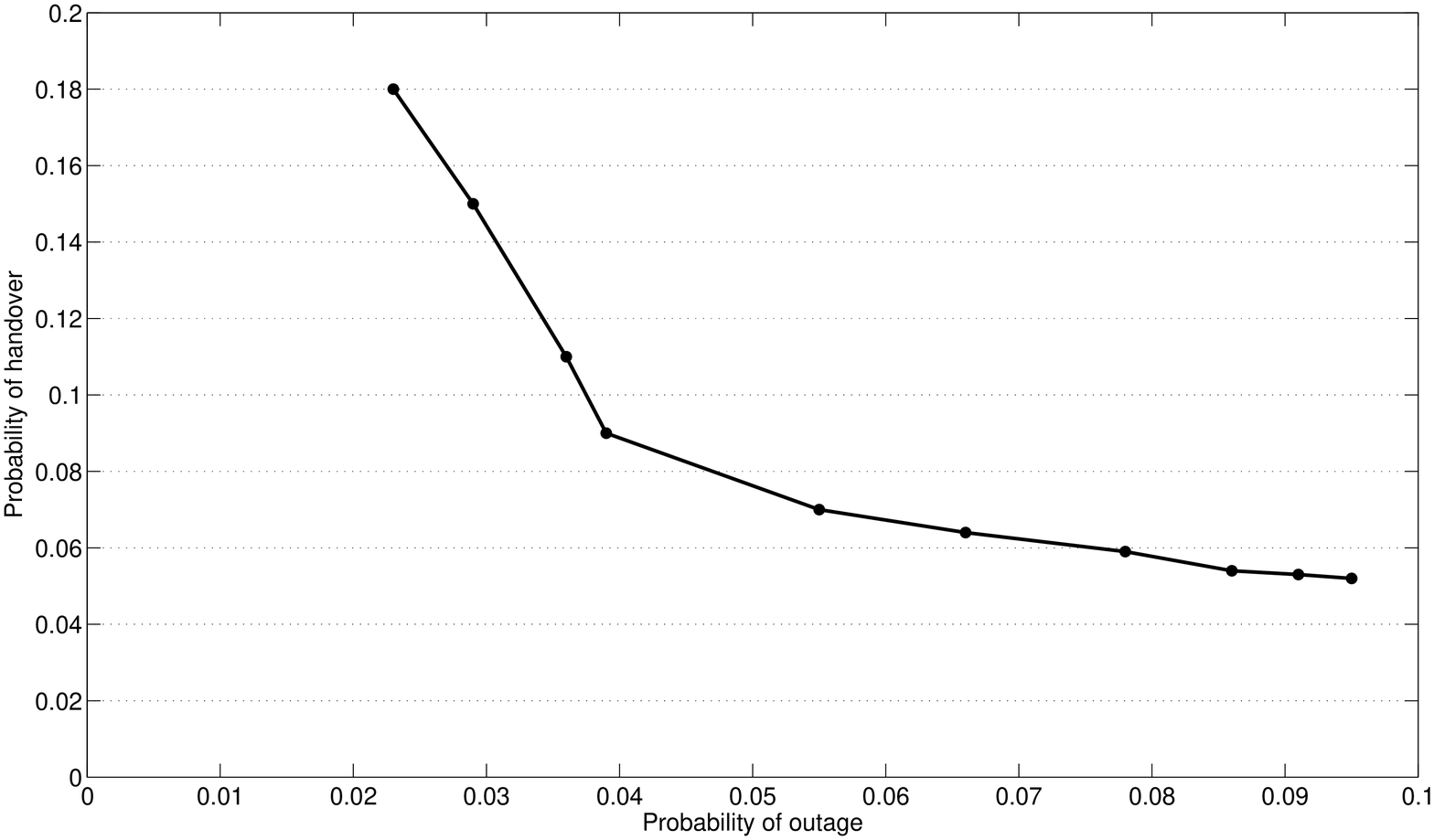}
  \caption{Handover and outage probability evaluation based on the optimal $h(n)$ values that resulted from optimization problem \eqref{objfunct}.}
  \label{fig:opth3a}
  \end{figure} 

In the next subsection, we focus on the case of a multi-cell system. 

\begin{table}
   \caption{Average number of handovers and outages. $h_{\rm opt 1}$ refers to the optimal solution of optimization problem~\eqref{eq:phstpo}, $h_{\rm opt 2}$ to \eqref{objfunc}, and $h_{\rm opt 3}$ to \eqref{objfunct}}.\vspace{-12mm}
  \begin{center}
  \begin{tabular}{ | c || c | c | c | }
    \hline
     & $v=5\:$m/s & $v=20\:$m/s & $v=40\:$m/s \\ \hline \hline
    $\bar{H} (h=0 \:{\rm dB})$ & $8.8$ & $12.4$ & $14.2$ \\ \hline 
    $\bar{O} (h=0 \:{\rm dB})$ & $11.4$ & $13.1$ & $14.6$ \\ \hline \hline
    $\bar{H} (h=2 \:{\rm dB})$ & $2.6$ & $3.5$ & $4.8$ \\ \hline
    $\bar{O} (h=2 \:{\rm dB})$ & $11.5$ & $13.6$ & $16.8$ \\ \hline \hline
    $\bar{H} (h=4 \:{\rm dB})$ & $2.2$ & $2.9$ & $3.8$ \\ \hline
    $\bar{O} (h=4 \:{\rm dB})$ & $15.2$ & $16.4$ & $18.5$ \\ \hline \hline
    $\bar{H} (h_{\rm opt 1})$ & $1.1$ & $1.5$ & $2.1$ \\ \hline
    $\bar{O} (h_{\rm opt 1})$ & $9.1$ & $10.2$ & $11.7$ \\ \hline \hline
    $\bar{H} (h_{\rm opt 2})$ & $2.1$ & $2.8$ & $3.5$ \\ \hline
    $\bar{O} (h_{\rm opt 2})$ & $5.3$ & $5.9$ & $6.5$ \\ \hline \hline
    $\bar{H} (h_{\rm opt 3})$ & $1.6$ & $2.4$ & $2.9$ \\ \hline
    $\bar{O} (h_{\rm opt 3})$ & $5.8$ & $6.6$ & $7.1$ \\
    \hline
  \end{tabular}
  \end{center}
  \label{table:average}
  \end{table}

\subsection{Performance in a Multi-cell System}

In this section we evaluate the proposed adaptive handover algorithm in a multi-cell environment consisting of $8$ hexagonal cells. We compare the results to those obtained by the method in~\cite{sant02}, which considers the same system set-up as in our study. Other related work, such as~\cite{akar03,AM_INFO2003}, cannot be used for comparison because the do not consider the wireless channel dynamics, or because they are applied for soft handover, whereas we are more focused on hard handover. 
 
The MT begins its trip close to BS$_0$ and moves towards the remaining cells in a straight line. The parameters of the simulation environment are the same as the ones assumed in the previous two-cells scenario (distance between the base stations, MT speed, etc.)

In order to get an intuition of the multi-cell system operations, we compute the average number of handovers and outages under different MT velocities ($v$) and hysteresis margins ($h$). The average number of outages and handovers is defined as
follows:
\begin{align*}
&\bar{O}=\sum_{n=1}^{N-1}P_O(n)\,, \qquad \bar{H}=\sum_{n=1}^{N-1}P_H(n)\,.
\end{align*}
The simulations results are summarized in Tab.~\ref{table:average}, when $v=5, 20, 40$~m/s and $h=0, 2, 4$~dB. We apply the optimal hysteresis margins resulting from optimization problems \eqref{eq:phstpo}, \eqref{objfunc} and \eqref{objfunct} and compare the results of the optimizations to the constant hysteresis margin proposed in~\cite{sant02}. 
Recall that~\eqref{eq:phstpo} guarantees minimum average handovers, while \eqref{objfunc} guarantees minimum average outages in the network. Eq.~\eqref{objfunct} manages the trade-off and guarantees balanced network operation in terms of both mean number of handovers and outages. From the table, we can conclude that our method substantially outperforms the one proposed in~\cite{sant02}. 

Finally, in the next subsection, we conclude the numerical investigations by applying and studying the accuracy of the proposed approximations. 

\subsection{Approximations Accuracy}

In this subsection we study the accuracy of the approximations presented in Subsections \ref{subsec:ulb1}, \ref{subsec:ulb2} , and \ref{subsec:ulb3}. Recall that the proposed bounds approximate the handover and outage probabilities. Our evaluation includes the execution of several simulations that give a general view of the behavior of the proposed approximations. 

Based on inspection of simulation experiments, we plot the results for few representative cases, namely: 1) with a medium value for $k$ in combination with a low value for $m$, 2) with a medium value for both $k$ and $m$, 3) with a large value for both $k$ and $m$, and 4) with a small value for both $k$ and $m$.

In Fig.~\ref{fig:Approximation1} we plot the handover probability based on the analytical value and the proposed approximations, with $k=6$ and $m=3$. The first method provides the best approximation of the handover probability. On the other hand, the computation complexity of the first method is high. Therefore, test results suggests that it is less expensive to adopt one of the remaining approximations that introduce less computational cost in the system. 

\begin{figure}[t]
  \centering
  \includegraphics[width=0.6 \textwidth]{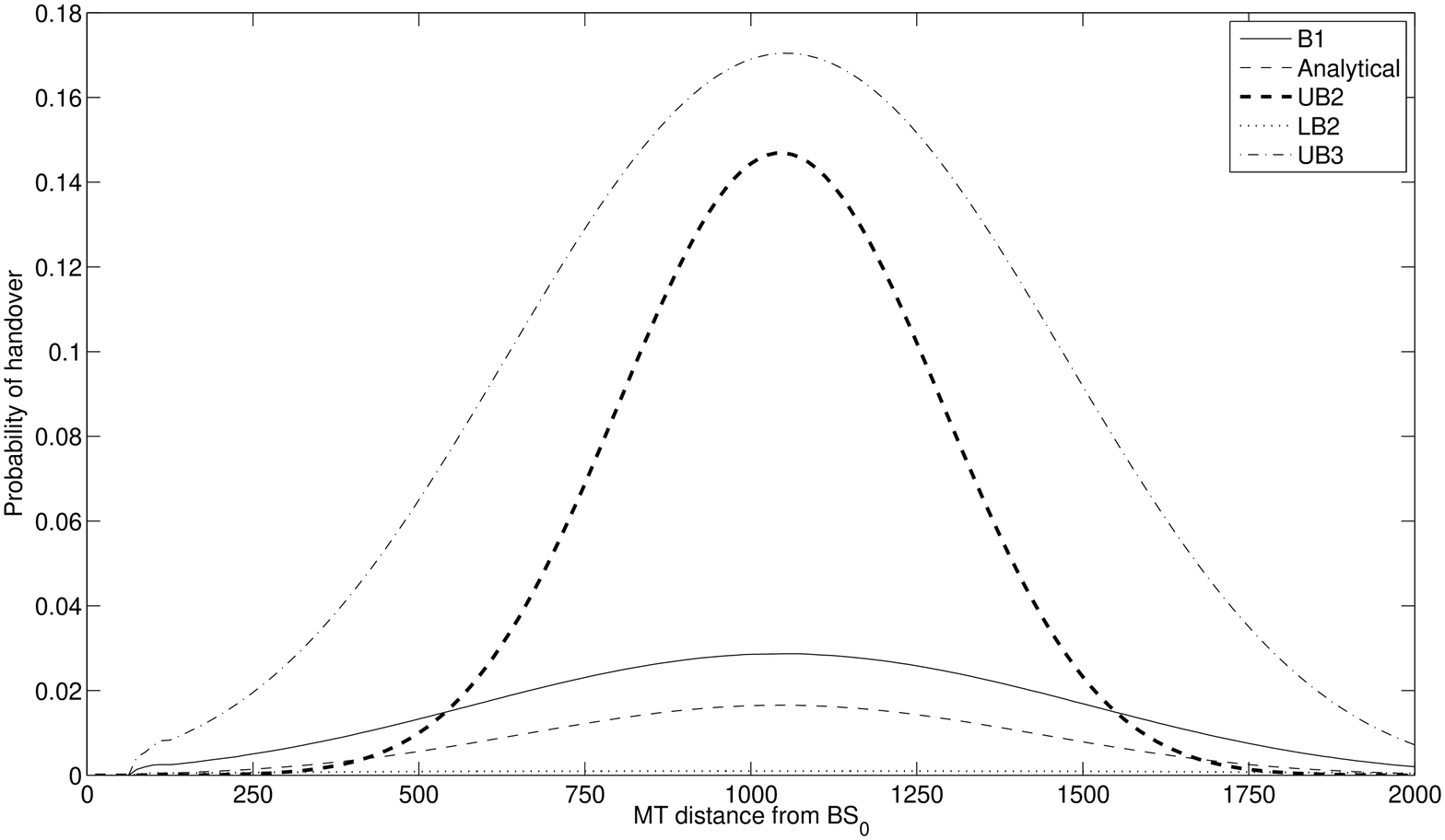}
  \caption{Approximation accuracy ($k=6$, $m=3$). B1 refers to the bound given by Approximation 1 of Subsection\ref{subsec:ulb1}, LB2 and UB2 refers to the lower and upper bounds of Approximation 2 of Subsection \ref{subsec:ulb2}, and UB3 refers to Approximation 3 of Subsection~\ref{subsec:ulb3}}
  \label{fig:Approximation1}
  \end{figure} 
  
\begin{figure}[t]
  \centering
  \includegraphics[width=0.6 \textwidth]{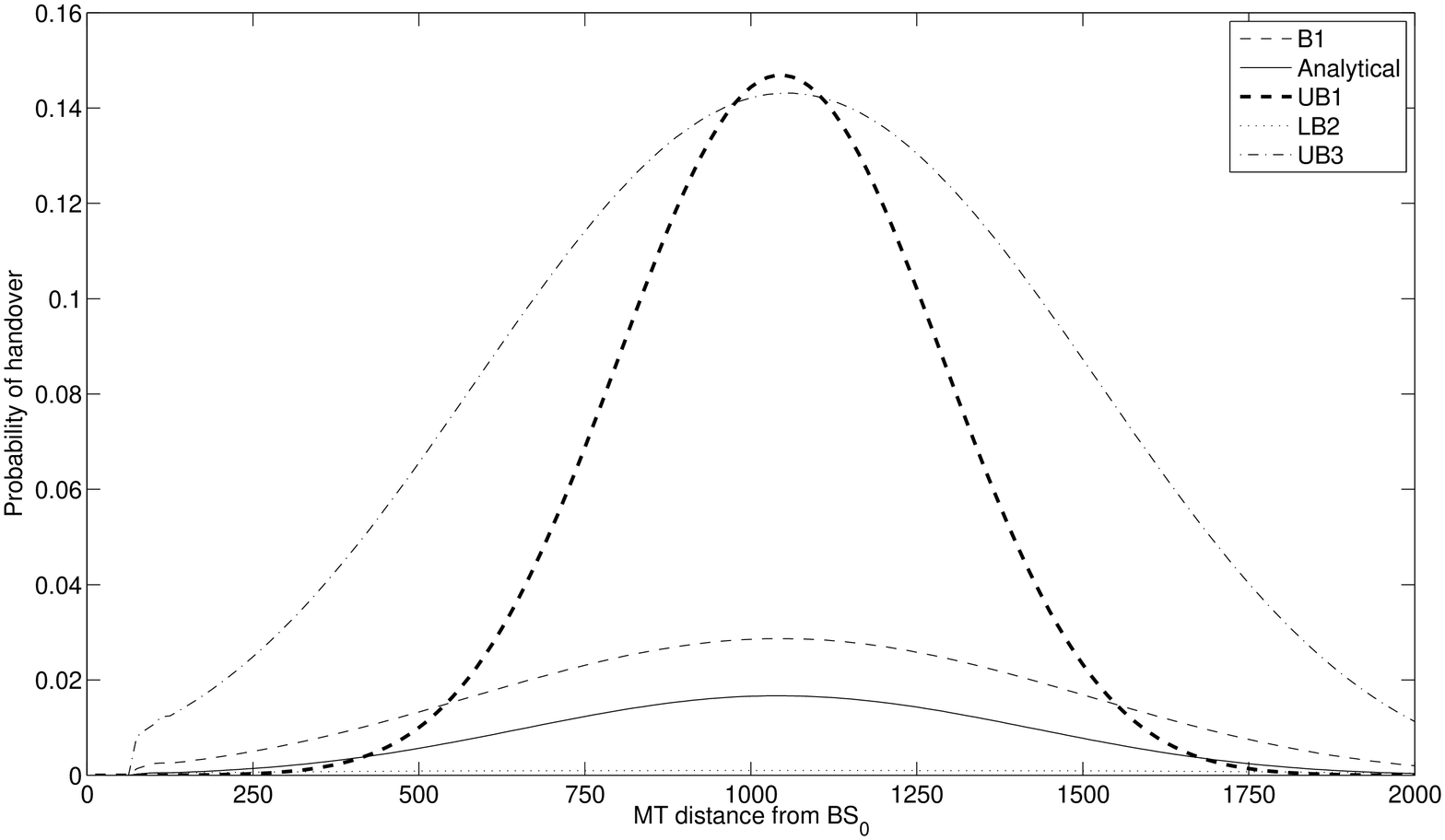}
  \caption{Approximation accuracy ($k=6$, $m=4$). B1 refers to the bound given by Approximation 1 of Subsection\ref{subsec:ulb1}, LB2 and UB2 refers to the lower and upper bounds of Approximation 2 of Subsection \ref{subsec:ulb2}, and UB3 refers to Approximation 3 of Subsection~\ref{subsec:ulb3}}
  \label{fig:Approximation2}
  \end{figure} 

Then we set $m=4$ and we progressively increase $k$ ($k=6$ and $k=8$). In both cases (Figs. \ref{fig:Approximation2} and \ref{fig:Approximation3}) the best approximation is given by the first methodology. An important observation is that the accuracy of the second and the third approximation is getting worse while $k$ is increasing. The first approximation is not affected.
 
  \begin{figure}[t]
  \centering
  \includegraphics[width=0.6 \textwidth]{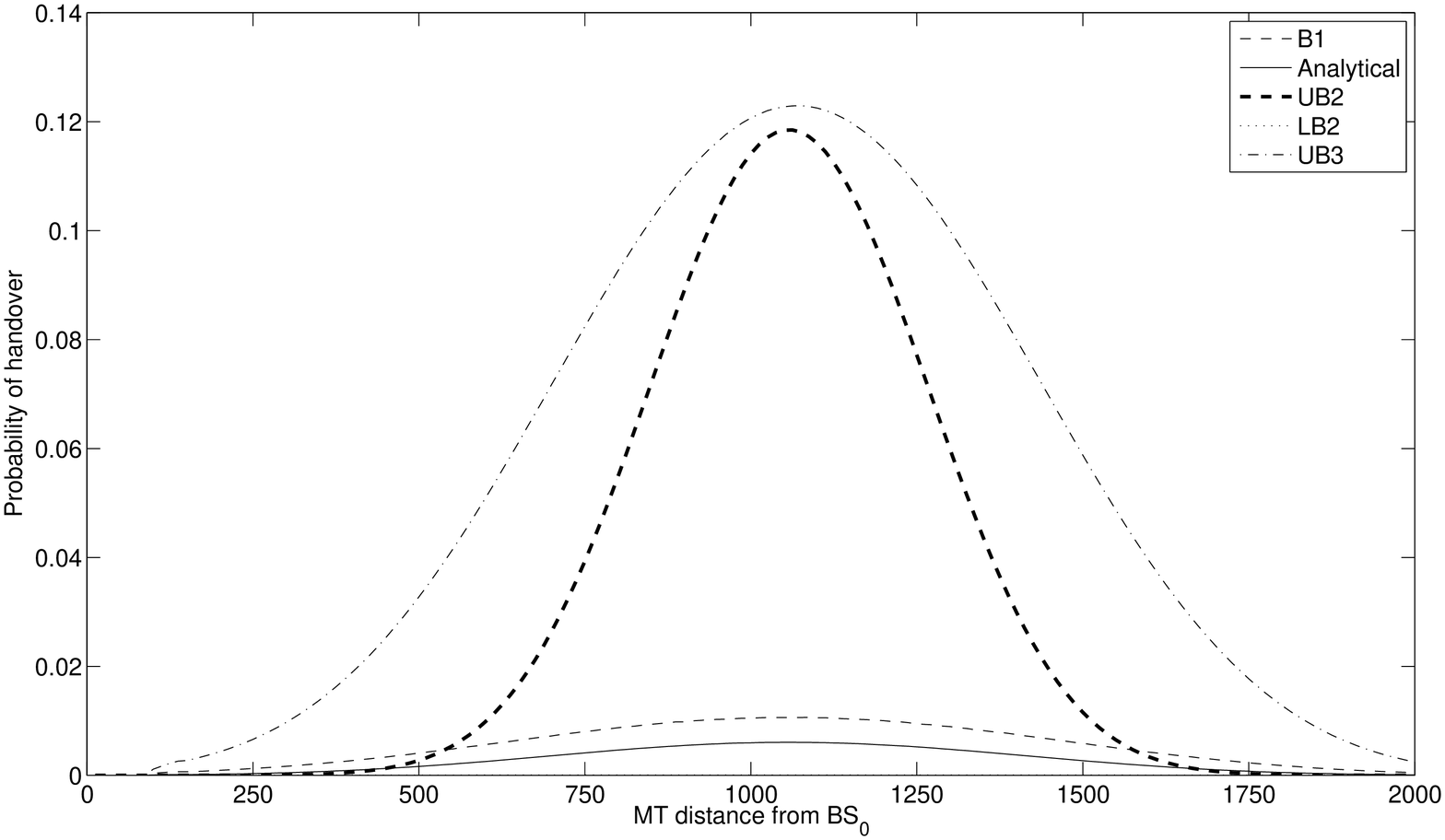}
  \caption{Approximation accuracy ($k=8$, $m=4$). B1 refers to the bound given by Approximation 1 of Subsection\ref{subsec:ulb1}, LB2 and UB2 refers to the lower and upper bounds of Approximation 2 of Subsection \ref{subsec:ulb2}, and UB3 refers to Approximation 3 of Subsection~\ref{subsec:ulb3}}
  \label{fig:Approximation3}
  \end{figure} 

The last scenario includes the adoption of small values for both $k$ and $m$ ($k=4$ and $m=3$). In that case (Fig. \ref{fig:Approximation4}), the first bound acts as a lower bound (not as an upper bound, like in the previous cases). Besides, the second approximation acts as an upper bound when the MT is located close to the cell boundaries and as a lower bound when the MT is located close to the base stations. 

\begin{figure}[t]
  \centering
  \includegraphics[width=0.6 \textwidth]{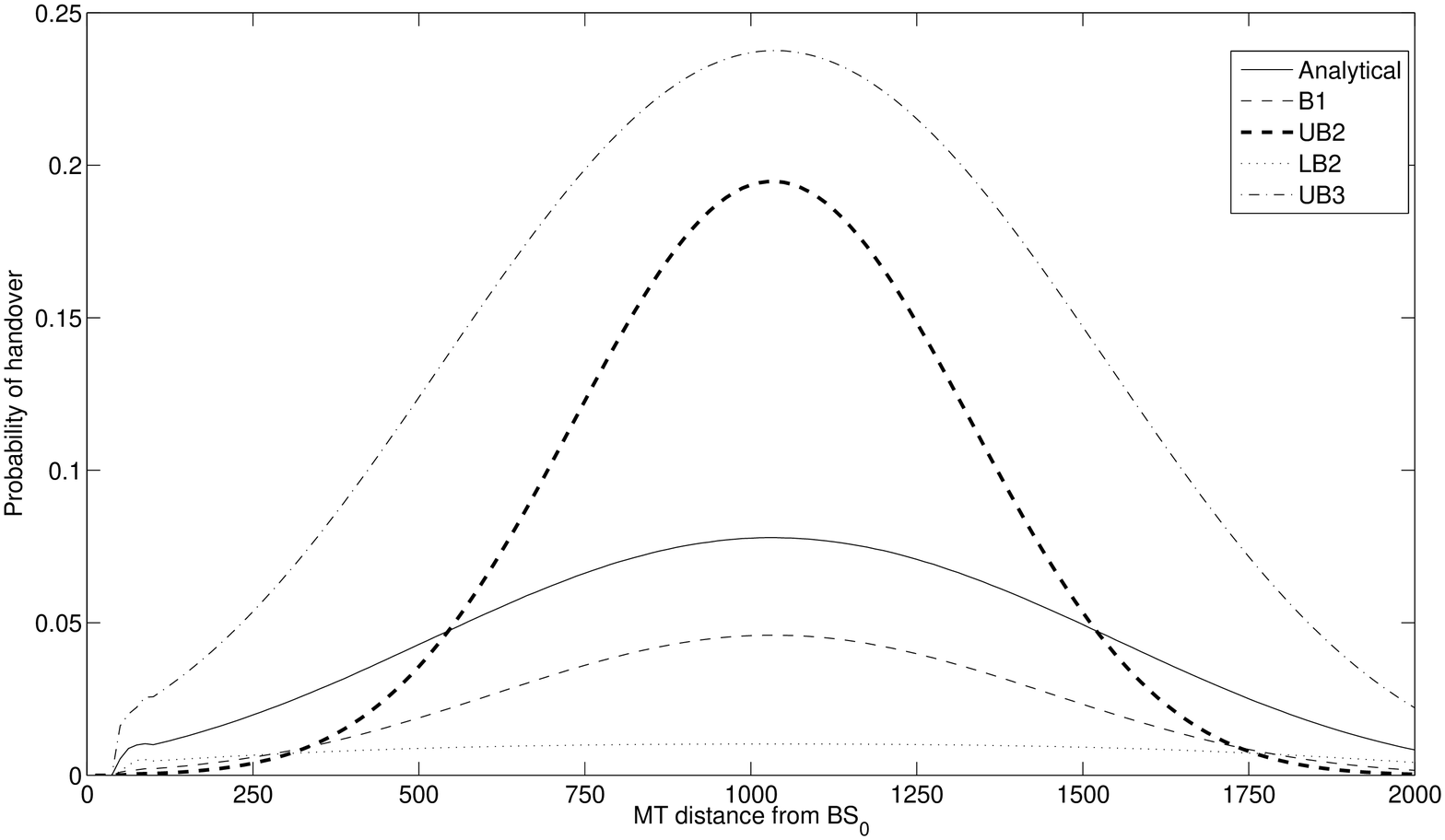}
  \caption{Approximation accuracy ($k=4$, $m=3$). B1 refers to the bound given by Approximation 1 of Subsection\ref{subsec:ulb1}, LB2 and UB2 refers to the lower and upper bounds of Approximation 2 of Subsection \ref{subsec:ulb2}, and UB3 refers to Approximation 3 of Subsection~\ref{subsec:ulb3}}
  \label{fig:Approximation4}
  \end{figure}   

\section{Conclusion} \label{sec:conclusions}

A hybrid system model of handover algorithm was presented. The performance indicators of the handover in terms of outage probability and handover probability were characterized together with approximations and upper and lower bounds of the probabilities. Then, based on such a characterization, some optimization strategies were proposed to optimally take the handover decision. Moreover, a solution algorithm of reduced computational complexity was developed to solve these problems. Monte Carlo simulations illustrated the proposed analysis for the case of two cells systems and multi-cell systems. In particular, it was shown that the analysis is accurate and that the proposed handover optimization outperforms existing methods in the literature.


\bibliographystyle{IEEEtran}
\bibliography{IEEEabrv,megabiblio.bib}{}

\end{document}